\documentclass{sig-alternate}

\usepackage{xcolor}
\usepackage[noend]{algpseudocode}                            \usepackage{algorithm}

\usepackage{subfigure}
\usepackage{enumerate}

\usepackage{xspace,amstext,amssymb,url}

\usepackage{stmaryrd} 
\usepackage{framed}

\usepackage{graphicx}
\usepackage{tikz}
\usepackage{paralist}
\usepackage{cite} 
\usetikzlibrary{automata}
\usetikzlibrary{snakes}
\usetikzlibrary{calc}

\newcommand{\nc}{\newcommand}
\nc{\cO}{{\mathcal O}}
\nc{\cD}{{\mathcal D}}
\nc{\cR}{{\mathcal R}} 
\nc{\cM}{{\mathcal M}}
\nc{\cV}{{\mathcal V}}
\nc{\cP}{{\mathcal P}}
\nc{\cC}{{\mathcal C}}
\nc{\cS}{{\mathcal S}}
\nc{\cVbot}{{\mathcal V}_{\bot}}
\nc{\nat}{\Bbb{N}}

\nc{\delim}{\text{Delim}}

\nc{\lab}{\text{lab}}
\nc{\true}{\textsf{true}}
\nc{\false}{\textsf{false}}

\nc{\wildcard}{\bullet}

\nc{\dimtab}{\text{dim}}
\nc{\nfa}{\text{NFA}\xspace}
\nc{\nfas}{\text{NFAs}\xspace}
\nc{\dfa}{\text{DFA}\xspace}
\nc{\dfas}{\text{DFAs}\xspace}

\nc{\re}{\text{RE}\xspace}
\nc{\res}{\text{REs}\xspace}

\renewcommand{\epsilon}{\varepsilon}

\newcommand{\sem}[1]{\llbracket #1 \rrbracket}
\nc{\numbercells}[2]{\ensuremath{|\cells{#1}{#2}|}}
\nc{\cells}[2]{\ensuremath{\sem{#1}_{#2}}}

\nc{\tdsl}{{\sc Sculpt}\xspace}
\nc{\chisel}{\tdsl}
\nc{\corechisel}{core-\chisel}
\nc{\TDSL}{\tdsl}
\nc{\tdsltitle}{SCULPT}

\nc{\expspace}{EXPSPACE\xspace}  
\nc{\exptime}{EXPTIME\xspace} 
\nc{\pspace}{PSPACE\xspace} 
\nc{\np}{NP\xspace} 
\nc{\conp}{coNP\xspace} 
\nc{\ptime}{PTIME\xspace} 
\nc{\fptime}{FPTIME\xspace} 
\nc{\nlogspace}{NLOGSPACE\xspace} 
\nc{\sharpp}{\#P\xspace}
\nc{\sharppspace}{\#PSPACE\xspace}
\nc{\sharpexpspace}{\#EXPSPACE\xspace}

\nc{\da}{\ensuremath{\downarrow}}
\nc{\ua}{\ensuremath{\uparrow}}
\nc{\la}{\ensuremath{\leftarrow}}
\nc{\ra}{\ensuremath{\rightarrow}}

\nc{\up}{\textsf{up}\xspace}
\nc{\down}{\textsf{down}\xspace}
\renewcommand{\left}{\textsf{left}\xspace}
\renewcommand{\right}{\textsf{right}\xspace}
\nc{\col}{\textsf{col}\xspace}
\nc{\row}{\textsf{row}\xspace}
\nc{\unique}{\textsf{unique}\xspace}
\nc{\uniqueperrow}{\textsf{unique-per-row}\xspace}

\newcommand{\alphabet}{\ensuremath{\Sigma}}
\newcommand{\tokens}{\ensuremath{\Delta}}
\DeclareMathOperator{\lang}{\ensuremath{\mathcal{L}}}
\DeclareMathOperator{\tokendef}{\ensuremath{\Theta}}
\DeclareMathOperator{\coords}{coords}

\newcommand{\regionmodels}{\models_{\text{region}}}
\newcommand{\ignore}[1]{}

\newtheorem{theorem}{Theorem}[section]

\newtheorem{proposition}[theorem]{Proposition}
\newtheorem{definition}[theorem]{Definition}
\newtheorem{example}[theorem]{Example}
\newtheorem{lemma}[theorem]{Lemma}
\newtheorem{remark}[theorem]{Remark}

\nc{\rootcoord}{\textsf{root}}

\newcommand{\newrow}{\ensuremath{\langle\textsf{new row}\rangle}}

\DeclareMathOperator{\level}{lvl}

\begin{document}

\title{\tdsltitle: a Schema Language for Tabular Data on the Web}

\numberofauthors{3}                     \author{
\alignauthor
Wim Martens\\
       \affaddr{Universit\"at Bayreuth}
       \email{\small{wim.martens@uni-bayreuth.de}}
\alignauthor
Frank Neven\\
       \affaddr{Hasselt University and transnational University of
         Limburg}
       \email{\small{frank.neven@uhasselt.be}}
\alignauthor
Stijn Vansummeren\\
       \affaddr{Universit\'e Libre de Bruxelles}
       \email{\small{stijn.vansummeren@ulb.ac.be}}
}

\additionalauthors{}

\maketitle

\begin{abstract}
Inspired by the recent working effort towards a recommendation by the 
World Wide Web Consortium (W3C) for tabular data and metadata on the Web,
we present in this paper a concept for a schema language for tabular
web data called \tdsl. The language consists of rules constraining and defining the structure of regions in the table. These regions are defined through the novel formalism of region selection expressions. We present a formal model for \tdsl and obtain a linear time combined complexity evaluation algorithm. In addition, we consider weak and strong streaming evaluation for \tdsl and present a \tdsl fragment for each of these streaming variants. 
Finally, we discuss several extensions of \tdsl including alternative semantics, types, complex content, and explore region selection expressions as a basis for a transformation language.
\end{abstract}

\makeatletter{}\section{Introduction}

Despite the availability of numerous standardized formats for
semi-structured and semantic web data such as XML, RDF, and JSON, a
very large percentage of data and open data published on the web,
remains tabular in nature.\footnote{\footnotesize Jeni Tennison, one
  of the two co-chairs of the W3C CSV on the Web working group claims
  that ``over 90\% of the data published on data.gov.uk is tabular
  data''~\cite{Jeni90}.}  Tabular data is most commonly published in
the form of comma separated values (CSV) files because such files are open and therefore processable by numerous
tools, and tailored for all sizes of files ranging from a number of
KBs to several TBs. Despite these advantages, working with CSV files
is often cumbersome because they are typically not accompanied by a
\emph{schema} that describes the file's structure (i.e., ``the second
column is of integer datatype'', ``columns are delimited by tabs'',
\dots) and captures its intended meaning. Such a description is
nevertheless vital for any user trying to interpret the file and
execute queries or make changes to it. In other data models, the
presence of a schema is also important for query optimization
(required for scalable query execution if the file is large), as well
as other static analysis tasks. Finally, we strongly believe that
schemas are a prerequisite for unlocking huge
amounts of tabular data to the Semantic Web. Indeed, unless we have a
satisfactory way of describing the structure of tabular data we cannot
specify how its contents should be interpreted as RDF. Drawing
parallels with relational databases, observe that R2RML
mappings~\cite{r2rml-spec} (the W3C standard for mapping relational
databases to RDF) inherently need to refer to the schema (structure)
of the relational database in order to specify how database tuples can
be represented as RDF. 

In recognition of this problem, the \emph{CSV on the Web} Working
Group of the World Wide Web Consortium \cite{w3cchartercsv} argues for
the introduction of a schema language for tabular data to ensure
higher interoperability when working with datasets using the CSV or
similar formats. In particular, their charter states
\cite{w3cchartercsv}: 
\begin{quote} \it Whether converted to other formats or not, there is
  a need to describe the content of the CSV file: its structure,
  datatypes used in a specific column, language used for text fields,
  access rights, provenance, etc. This means that metadata should be
  available for the dataset, relying on standard vocabulary terms, and
  giving the necessary information for applications. The metadata can
  also be used for the conversion of the CSV content to other formats
  like RDF or JSON, it can enable automated loading of the data as
  objects, or it can provide additional information that search
  engines may use to gain a better understanding of the content of the
  data.
\end{quote}
In the present paper, we introduce \tdsl as a concept
for such a schema language for tabular data.\footnote{The name \tdsl for the language is in honour of Michelangelo, who
allegedly said \emph{``Every block of stone has a statue inside it and it is
the task of the sculptor to discover it.''} Readers who like acronyms
can read \tdsl as \textbf{SC}hema for \textbf{U}n-\textbf{L}ocking and
\textbf{P}rocessing \textbf{T}abular data.}

The critical reader may wonder whether designing such a schema
language isn't trivial. After all, doesn't it suffice to be able to
specify, for each column, the column's name and the type of data
allowed in its cells---similar to how relational database schemas
are defined using the SQL data definition language? The answer is
no. The reason is that there is a lot of variation in the tabular data
available on the web and that there are examples abound of tabular
data whose structure cannot be described by simple rules of the form
``column $x$ has datatype $y$''. Figures~\ref{ex:1:csv},
\ref{ex:2:csv}, and \ref{fig:PDB}, for example, show some tabular data
sets drawn from the Use Cases and Requirements document drafted by the
W3C CSV on the Web working group~\cite{w3c-csv-usecases}. Notice how,
in contrast to ``standard'' CSV files, Figure~\ref{ex:1:csv} has a
header consisting of multiple lines. This causes the data in the first
column to be non-uniform. Further notice how the \texttt{provenance}
data in the Figure~\ref{ex:2:csv} is spread among multiple
columns. Finally, notice how the shape of the rows in
Figure~\ref{fig:PDB} depends on the label in the first column of the
column: \texttt{TITLE} rows have different structure than
\texttt{AUTHOR} rows, which have a different structure than
\texttt{ATOM} rows, and so on. 

\tdsl schemas use the following idea to
describe the structure of these tables.
At their core, \tdsl schemas consist of rules of the form $\varphi \to
\rho$. Here, $\varphi$ selects a \emph{region} in the input table
(i.e., a subset of the table's cells) and $\rho$ constrains the
allowed structure and content of this region. A table is valid with
respect to a \tdsl schema if, for each rule in the schema, the region
selected by $\varphi$ satisfies the content constraints specified by
$\rho$.  It is important to note that \tdsl's expressive power goes
well beyond that of classical relational database schemas since 
\tdsl's region selectors are not limited to selecting columns. In
particular, the language that we propose for selecting regions is
capable of navigating through a table's cells bears much resemblance
to the way XPath~\cite{xpath2} navigates through the nodes of
an XML tree. For tokenizing the content of single cells, we draw inspiration from
XML Schema simple types (\cite{xsd-0}, Section
2.2).  Both features combined will allow us to express the use
cases of the W3C CSV on the Web Working Group.

We note that the W3C is also working on a schema language for tabular
data~\cite{w3c-metadata-vocabulary}. At the moment, however, that schema
language focuses on orthogonal issues like describing, for instance,
\emph{datatypes} and \emph{parsing cells}. Also, it only provides facilities
for the selection of \emph{columns}, and is hence not able to
express the schema of the more advanced use cases. \tdsl, in contrast,
draws inspiration from well-established theoretical tools from logic
and formal languages, which adds to the robustness of our approach.
Due to the above mentioned orthogonality we expect that it is not
difficult to integrate ideas from this paper in the W3C proposal.

\smallskip \noindent
In summary, we make the following contributions.
\begin{compactenum}[1.]
\item We illustrate the power of \chisel, and its suitability as a
  schema language for tabular data on the web, by expressing several
  use cases drafted by the CSV on the Web W3C working
  group~\cite{w3c-csv-usecases}. (Section~\ref{sec:chisel:examples}) 
\item We provide a formal model for the core of \chisel. A key
  contribution in this respect is the introduction of the region
  selector language.  (Section~\ref{sec:formalmodel}) 
\item We show that, despite its rather attractive expressiveness,
  tables can be efficiently validated w.r.t.\ \chisel schemas. In
  particular, when the table is small enough to be materialized in
  main memory, we show that validation can be done in linear time
  combined complexity (Section~\ref{sec:evaluation:linear}). For scenarios
  where materialization in main memory is not possible, we consider
  the scenario of streaming (i.e., incremental) validation. We formally
  introduce two versions of streaming validation: \emph{weak
    streamability} and \emph{strong streamability}. (Their differences
  are described in detail in Section~\ref{sec:streaming-evaluation}.)
  We show in particular that the fragment of \corechisel where region
  selectors can only look ``forward'' and never ``backward'' in the CSV
  file is \emph{weakly streamable}. If we further restrict region
  selectors to be both forward-looking and \emph{guarded} (a notion
  formalized in Section~\ref{sec:streaming-evaluation}) validation
  becomes \emph{strongly streamable}. All of the W3C Working group use
  cases considered here can be expressed using forward and guarded
  region selectors, hence illustrating the practical usefulness of
  this fragment.
\item While our focus in this paper is on introducing \chisel as a
  means for specifying the structure of CSV files and related formats,
  we strongly believe that region selector expressions are a
  fundamental component in developing other features mentioned in the
  charter of the W3C CSV on the Web Working Group, such as a CSV
  transformation language (for converting tabular data into other
  formats such as RDF or JSON), the specification of the language used
  for text fields; access rights; provenance; etc. While a full
  specification of these features is out of this paper's scope, we
  illustrate by means of example how \chisel could be extended to
  incorporate them.  (Section~\ref{sec:extensions})
\end{compactenum}

\smallskip
\noindent {\bf Note.} Due to space restrictions, proofs of formal statements are only
sketched. Proofs are provided in the Appendix.

\smallskip
\noindent {\bf Related Work.} The present paper fits in the line of
research, historically often published in the WWW conference, that
aims to formalize and study the properties of various W3C working
group drafts and standards (including XML
Schema~\cite{BexGNV-www08,DBLP:conf/www/BexMNS05},
SPARQL~\cite{DBLP:conf/www/ArenasCP12,DBLP:conf/pods/LosemannM12,DBLP:journals/tods/PerezAG09},
and RDF~\cite{DBLP:conf/www/RymanHS13,shape-expr}) with the aim of
providing feedback and input to the working group's activities. 

Given the numerous benefits of schemas for data processing, there is a
large body of work on the development, expressiveness, and properties
of schema languages for virtually all data models, including the
relational data model,
XML~\cite{BexGNV-www08,DBLP:conf/www/BexMNS05,DBLP:journals/pvldb/MartensNNS12,MartensNSB-tods06,DBLP:journals/jcss/GeladeN11},
and, more recently,
RDF~\cite{DBLP:conf/www/RymanHS13,shape-expr}. \tdsl differs from the
schema languages considered for XML and RDF in that it is specifically
designed for tabular data, not tree-structured or graph-structured
data. Nevertheless, the rule-based nature of \tdsl draws inspiration
from our prior work on rule-based and pattern-based schema languages
for
XML~\cite{DBLP:journals/pvldb/MartensNNS12,MartensNSB-tods06,DBLP:journals/jcss/GeladeN11}.

As already mentioned, while traditional relational database schemas
(formulated in e.g., the SQL data definition language) are
specifically designed for tabular data, they are strictly less
expressive than \tdsl schemas in the sense that relational schemas
limit region selection expressions to those that select columns
only. A similar remark holds for other recent proposals of CSV
schema languages, including the CSV Schema language proposed by the UK
National Archives~\cite{csv-schema-national-archives}, and Tabular
Data Package~\cite{tabular-data-package}. The remark also applies to
the part of Google's Dataset Publishing Language
(DSPL)~\cite{google-dspl} describing the contents of CSV files. In
contrast, DSPL also has features to relate data from multiple CSV
files, which \tdsl does not yet have.

The problem of streaming schema validation has been investigated in
the XML context for DTDs and XML schemas \cite{SegoufinV-pods02,
  SegoufinS-icdt07, MartensNSB-tods06,DBLP:conf/www/KumarMV07}. In
this work, the focus is on finding algorithms that can validate an XML
document in a single pass using constant memory or, if this is not
possible, a memory that is bounded by the depth of the document. Our
notion of streaming, in contrast, is one where we can use a memory
that is not constant but at most logarithmic in the size of the table
(for strong streaming), or at most linear in the number of columns and
logarithmic in the number of rows (for weak streaming). This allows us
to restrict memory when going from one row to the next and is essential to be able to
navigate downwards in \tdsl region selection expressions.

While streaming validation is undoubtedly an important topic for all of
the CSV schema languages mentioned
above~\cite{csv-schema-national-archives,tabular-data-package,google-dspl}
(the National Archives Schema Language mentions it as an explicit
design goal), no formal streaming validation algorithm has been
proposed for them, to the best of our knowledge.

\makeatletter{}\section{\tdsltitle\ By Example}
\label{sec:chisel:examples}
\label{sec:examples}

In this section, we introduce \tdsl through a number of examples. The
formal semantics of the examples is defined in
Section~\ref{sec:formalmodel}. The syntax we use here is tuned for
making the examples accessible to readers and is, of course, flexible.

 \tdsl schemas operate on \emph{tabular documents}, which are
text files describing tabular data. \tdsl schemas consist of two parts
(cf.\ Figure~\ref{fig:climate-schema}). The first part, \emph{parsing
  information}, defines the row and column delimiters and further
describes how words should be tokenized. This allows to parse the text
file and build a table-like structure consisting of rows and
columns. In this section we allow some rows to have fewer columns than
others but we require them to be aligned to the left. That is,
non-empty rows always have a first column.  The second part of the
schema consists of \emph{rules} that interpret the table defined by
the first part as a
rectangular grid and enforce structure. In particular, rules are of
the form $\varphi\to\rho$, where $\varphi$ selects a \emph{region}
consisting of cells in the grid
while $\rho$ is a regular expression constraining the 
content of the selected region. We utilize a so-called \emph{row-based} semantics:
every row in the region selected by
$\varphi$ should be of a form allowed by $\rho$.
We refer to $\varphi$ as the
\emph{selector expression} and to $\rho$ as the \emph{content expression}.

Next, we illustrate the features of the language by means of
examples. All examples are inspired by the use cases and
requirements drafted by the CSV on the Web W3C working
group~\cite{w3c-csv-usecases}.

\begin{example}\upshape\label{ex:simple}
    Figure~\ref{fig:climate} contains a slightly altered fragment (we use a comma as a column separator) of a
  CSV file mentioned in Use Case 3, \emph{``Creation of consolidated
    global land surface temperature climate
    databank''}~\cite{w3c-csv-usecases}. 
The \tdsl schema, displayed as Figure~\ref{fig:climate-schema}, starts by describing
parsing information indicating that the column delimiter is a comma while the
row delimiter is a newline. Lines starting with a \%-sign are comments. Tokens are defined based on regular
expressions (regex for short).\footnote{For ease
  of exposition, we adopt the concise regex syntax popularized by
  scripting languages such as Perl, Python, and
  Ruby~\cite{mastering-regex} in all of our examples.}  For instance,
anything that matches the regex \texttt{[0-9]\{4\}"."[0-9]\{2\}}
follows the format \emph{four digits, dot, two digits}, and is
interpreted by the token \texttt{Timestamp} in the rules of the
schema (similar for {\tt Temperature}). Notice that we keep the
regexes short (and sometimes imprecise) for readability, but they
can of course be made arbitrarily precise if desired.

\begin{figure}[!h]
  \begin{framed}
\vspace{-1mm}
\begin{verbatim}
       ,   ARUA,  BOMBO, ENTEBBE AIR
1935.04, -99.00, -99.00,       27.83
1935.12, -99.00, -99.00,       25.72
1935.21, -99.00, -99.00,       26.44
1935.29, -99.00, -99.00,       25.72
1935.37, -99.00, -99.00,       24.61
1935.46, -99.00, -99.00,       24.33
1935.54, -99.00, -99.00,       24.89
\end{verbatim}
\vspace{-.5cm}
  \end{framed}
  \caption{Example tabular data inspired by Use Case 3 in \cite{w3c-csv-usecases}.\label{fig:climate}}
\end{figure}

\begin{figure}[!h]
\begin{framed}
\vspace{-2mm}
\begin{verbatim}
%  Parsing information
%% Delimiters
Col Delim = ,
Row Delim = \n

%% Tokens
%% left: token name 
%% right: regex

Timestamp = [0-9]{4}"."[0-9]{2}
Temperature = (-)?[0-9]{2}"."[0-9]{2}
ARUA = ARUA
BOMBO = BOMBO
ENTEBBE AIR = ENTEBBE AIR

% Rules

row(1) -> Empty, ARUA, BOMBO, ENTEBBE AIR
col(1) -> Empty | Timestamp
col(ARUA) -> Temperature
col(BOMBO) -> Temperature
col(ENTEBBE AIR) -> Temperature
\end{verbatim}
\vspace{-.5cm}
\end{framed}
\caption{Schema for tabular data of the type in Figure~\ref{fig:climate}.
  \label{fig:climate-schema}}
\end{figure}

All the XML Schema primitive types like \textsf{xs:integer},
\textsf{xs:string}, \textsf{xs:date}, etc are pre-defined as tokens in
a \tdsl schema. There is also a special pre-defined token {\tt Empty}
to denote that a certain cell is empty.

Notice that the schema in Figure~\ref{fig:climate-schema} has three
token definitions in which the regex defines only one character
sequence (namely: \texttt{AURA}, \texttt{BOMBO}, \texttt{ENTEBBE
  AIR}). In the sequel, we will omit such rules for reasons of
parsimony. For the same reason, we omit the explicit definition of
column and row delimiters when they are a comma and newline character,
respectively.

  The rule 
        
  \centerline{\texttt{row(1) -> Empty, ARUA, BOMBO, ENTEBBE AIR}}

  \noindent selects all cells in the first row and requires that the first is
  empty, the second contains {\tt ARUA}, the third {\tt BOMBO}, and
  the fourth {\tt ENTEBBE AIR}. Next, \texttt{col(1)} selects the
  region consisting of all cells in the first column. As \chisel
  assumes a row-based semantics per default,\footnote{We discuss an
    extension in Section~\ref{sec:extensions}.} the rule

  \centerline{\texttt{col(1) -> Empty | Timestamp}}

  \noindent 
  requires that every row in the selected region (notice that each such row consists of a single cell) is  
 either empty ({\tt Empty}) or contains data that matches the {\tt
    Timestamp} token. The expression \texttt{col(AURA)} selects
    all cells in the column below the cell containing {\tt ARUA}.
    The rule

       \centerline{\texttt{col(ARUA) -> Temperature}}
    
\noindent
therefore requires that every row in the selected region matches the {\tt Temperature} token. The two remaining rules
are analogous. The fragment in Figure~\ref{fig:climate}
satisfies the schema of Figure~\ref{fig:climate-schema}. \hfill $\blacksquare$
\end{example}

Before moving on to some more advanced examples, we discuss in more
detail the semantics of selector and content expressions.
Each cell in a table is identified by its \emph{coordinate}, which is
a pair $(k,\ell)$ where $k$ indicates the row number ($k \geq 1$) and
$\ell$ the column number ($\ell \geq 1$). In each rule $\varphi \to
\rho$, the selector expression $\varphi$ returns a {\emph{set}} of
coordinates (a region) and $\rho$ is a regular expression defining
the allowed structure of each row in the region selected by
$\varphi$. It is important to note that in each such row only the cells
which are selected by $\varphi$ are considered. Another way to
interpret the row-based semantics is that of a `group by' on the
selected region per row.

  \begin{figure}[t]
    \begin{framed}
      \vspace{-2mm}
\begin{verbatim}
QS601EW
Economic activity
27/03/2011

         ,        , Count   , Count
         ,        , Person  , Person     
         ,        , Activity, Activity
GeoID    , GeoArea, All     , Part-time
E92000001, England, 38881374, 27183134
W92000004, Wales  , 2245166 , 1476735
\end{verbatim}
      \vspace{-.5cm}
    \end{framed}
    \caption{Fragment of a CSV-like-file, inspired by Use Case 2 in \cite{w3c-csv-usecases}.\label{ex:1:csv}}
  \end{figure}

  \begin{figure}[t]
    \begin{framed}
      \vspace{-2mm}
\begin{verbatim}
%% Tokens
%% left: token name 
%% right: regex

name = QS[0-9]*EW
ctype = Economic Activity
geo_id = E[0-9]*

% Rules

row(1) -> name
row(2) -> ctype
row(3) -> Date
row(4) -> Empty
row(5) -> Empty, Empty, Count*
row(6) -> Empty, Empty, Person*
row(7) -> Empty, Empty, Activity*
row(8) -> GeoID, GeoArea, String*
col(GeoID) -> geo_id
col(GeoArea) -> String
down+(right+(GeoArea)) -> Number*
\end{verbatim}
      \vspace{-.5cm}
    \end{framed}
    \caption{``Schema'' for files of the type in Figure~\ref{ex:1:csv}.
      \label{ex:1:schema}}
  \end{figure}

The last rule we discussed in Example~\ref{ex:simple} uses a symbolic
coordinate {\tt ARUA} in its selector expression.  Its semantics is as
follows: a token $\tau$ returns the set of all coordinates $(k,\ell)$
whose cell contents matches $\tau$.  The operator {\tt row} applied to
a coordinate $(k,\ell)$ returns the set of coordinates
$\{(k,\ell')\mid \ell'>\ell\}$. This corresponds to the row consisting
of all elements to the right of $(k,\ell)$.  Note that coordinate
$(k,\ell)$ itself is not included. Applying {\tt row} to a \emph{set}
$S$ of coordinates amounts to taking the union of all {\tt
  row}$((k,\ell))$ where $(k,\ell) \in S$.  Similarly, the operator
{\tt col} applied to $S$ returns the union of the regions
$\{(k',\ell)\mid k'>k\}$ for each $(k,\ell)$ in $S$, corresponding to
columns below elements in $S$.
The selector expressions {\tt row(1)} and {\tt col(1)} that select the
``first row'' and ``first column'', respectively, use syntactic
sugar to improve readability. Formally, the notation {\tt row(k)} and
{\tt col(l)} abbreviate {\tt row($\{(k,0)\}$)} and {\tt col($\{(0,\ell)\}$)},
respectively. Using the same principle as above, this means that {\tt
  row(k)} selects the cells $\{(k,\ell) \mid \ell > 0\}$ and {\tt
  col(l)} selects $\{(k,\ell) \mid k > 0\}$. Notice that we use the
convention that the top left coordinate in tabular data bears the
coordinate $(1,1)$ --- for \emph{first row}, \emph{first column}. While
the value 0 does not refer to any cell in the table, it is used to define the
semantics of expressions.

The next example illustrates the use of slightly more complex
expressions for navigation and content.
\begin{example}\upshape\label{ex:medium}
  Figure~\ref{ex:1:csv} displays a (slightly altered) fragment of a
  CSV-like-file inspired by Use Case 2 (\emph{``Publication of National
  Statistics''}) in \cite{w3c-csv-usecases}.
  This fragment originates from the Office for National Statistics
  (UK) and refers to the dataset ``QS601EW Economic activity'' derived
  from the 2011 Census. The file starts with three lines of metadata,
  referring to the name of the file and the census date, continues
  with a blank line, before listing the actual data separated by
  commas.   Notice that this file is, strictly speaking, not a
  comma-separated-value file because not all rows have an equal number
  of columns.\footnote{Actually CSV does not have a standard, but the
    informative memo RFC4180 (http://tools.ietf.org/html/rfc4180)
    states rectangularity in paragraph 2.4.} Indeed, the first four
  rows have only (at most) one column and the later rows have four columns.
 Figure~\ref{ex:1:schema} depicts the \TDSL schema describing
  the structure of such tables.

  \begin{figure*}[t!]
    \vspace{-5mm}
    \begin{framed}
\begin{verbatim}
subject  predicate      object  provenance
:e4      type           PER
:e4      mention       "Bart"   D00124 283-286
:e4      mention       "JoJo"   D00124 145-149 0.9
:e4      per:siblings  :e7      D00124 283-286 173-179 274-281
:e4      per:age       "10"     D00124 180-181 173-179 182-191 0.9
:e4      per:parent    :e9      D00124 180-181 381-380 399-406 D00101 220-225 230-233 201-210
\end{verbatim}
      \vspace{-5mm}
    \end{framed}
    \caption{Fragment of a CSV-like file, inspired by Use Case 13 in \cite{w3c-csv-usecases}.\label{ex:2:csv}}
  \end{figure*}

  The schema starts by describing parsing information, analogous to
  Example~\ref{ex:simple}.  The first four rules are very basic and
  are similar to those of Example~\ref{ex:simple}.
                    We first describe the fifth rule of the
  schema: 

    \centerline{\texttt{row(5) -> Empty, Empty, Count*}}
  
  \noindent
  selects all cells in the fifth row, requiring the first two to be
  {\tt Empty} and the remaining non-empty cells to contain {\tt
    Count}. We note that the original data fragment from
  \cite{w3c-csv-usecases} contains 16 such columns. The remaining
  rules constraining rows are similar.

          The rule
    \texttt{col(GeoID) -> geo\_id}
  selects all cells below cells containing the word {\tt GeoID}. 
   The content
  expression says that this column contains values that match the {\tt
    geo\_id} token. 
          The last rule is the most interesting one:

    \centerline{\texttt{down+(right+(GeoArea)) -> Number*}.}
    
  \noindent
  This rule selects all cells appearing strictly downward and to the right of {\tt
    GeoArea} and requires them to be of type {\tt Number}. 
    More precisely, {\tt GeoArea} is a symbolic coordinate 
    selecting all cells containing the word {\tt GeoArea}.
    The navigational operators {\tt right} and {\tt down} select cells one step
    to the right and one step down, respectively, from a given coordinate. The operator {\tt +} indicates an arbitrary strictly positive number of applications of the navigational operator to which it is applied. In particular, as on the table given in Figure~\ref{ex:1:csv},  {\tt GeoArea} is the singleton cell with
  coordinate $(8,2)$, {\tt right(GeoArea)} returns $\{(8,3)\}$, while
 {\tt right+(GeoArea)} is the region $\{(8,\ell) \mid \ell >
  2\}$. Likewise, {\tt down(right+(GeoArea))} is the region $\{(9,\ell)
  \mid \ell > 2\}$ and, finally, {\tt down+(right+(GeoArea))} is the
  region downward and to the right of the {\tt GeoArea} coordinate, that is,
  $\{(k,\ell) \mid k > 8$ and $\ell > 2\}$. \hfill $\blacksquare$
        \end{example}

Example~\ref{ex:medium} uses more refined navigation than just
selecting a row or a column. \TDSL has four navigational axes: {\tt
  up, down, left, right} which navigate one cell
upward, downward, leftward, or rightward. These axes can be applied to
a set $S$ of coordinates and add a vector $v$ to it. More formally,
an axis $A$, when applied to a set $S$ of coordinates, returns $A(S)
:= \{c+v_A \mid c \in S\}$. Here,
\begin{compactitem}
\item $v_A = (-1,0)$ when $A = {\tt up}$, 
\item $v_A = (1,0)$ when $A = {\tt down}$.
\item $v_A = (0,1)$ when $A = {\tt right}$, and
\item $v_A = (0,-1)$ when $A = {\tt left}$.
\end{compactitem}
Furthermore, there is also an axis {\tt cell} that does not navigate
away from the current cells, i.e., ${\tt cell}(S) = S$. When applying an axis to a set of coordinates, we always return only the coordinates that are valid coordinates in the table. For example,
{\tt left(\{1,1\})} returns the empty set because $(1,0)$ is not a
cell in the table.

While the just discussed features of \tdsl are sufficient to describe
the structure of almost all CSV-like data on the Web Working group use
cases~\cite{w3c-csv-usecases}, we extend in
Section~\ref{sec:formalmodel} \tdsl to include XPath-like
navigation. These features will be useful for annotations and
transformations, see Section~\ref{sec:extensions}.
We now showcase \tdsl by illustrating it on the most challenging of
the W3C use cases.

\begin{example}\upshape
  Figure~\ref{ex:2:csv} contains a fragment of a CSV-like file, inspired by
  Use Case 13 in
  \cite{w3c-csv-usecases} (\emph{``Representing Entities and Facts Extracted From Text''}). Figure~\ref{ex:2:schema} depicts the
  \tdsl schema.   Compared to the previous examples, the most
  interesting rule is

  \quad \texttt{down+(right*(provenance))}

  \hfill \texttt{-> (prov-book, prov-pos*, prov-node?)*}

  \noindent 
   which states that every row that starts with a coordinate of the
   form $(k,4)$ (as {\tt provenance} only occurs in column 4) with $k > 1$ should match \texttt{(prov-book,
     prov-pos*, prov-node?)*}. Notice that the empty row starting at
   $(2,4)$ also matches this expression. Here, {\tt *} denotes an arbitrary number (including zero) of applications of the navigational operator. \hfill $\blacksquare$
\end{example}

  \begin{figure}
    \begin{framed}
      \vspace{-2mm}
\begin{verbatim}
% Tokens 
%% left: token name 
%% right: regex

rdf-id    = [a-zA-Z0-9]*:[a-zA-Z0-9]*
rdf-lit   = "[a-zA-Z0-9]*"
prov-book = D[0-9]{5}
prov-pos  = [0-9]{3}-[0-9]{3}
prov-node = [0-9].[0.9]
word      = [a-z]*

% Rules
row(1) -> subject,predicate,object, provenance
col(subject)   -> rdf-id
col(predicate) -> word | rdf-id
col(object)    -> rdf-lit | rdf-id
down+(right*(provenance)) 
         -> (prov-book, prov-pos*, prov-node?)*
\end{verbatim}
      \vspace{-5mm}
    \end{framed}
    \caption{Schema for files of the type in Figure~\ref{ex:2:csv}.\label{ex:2:schema}}
  \end{figure}

\makeatletter{}
\section{Formal Model}
\label{sec:formalmodel}

In this section, we present a formal model for the logical core of
\tdsl. We refer to this core as \corechisel and discuss extensions
in Section~\ref{sec:extensions}. We first define the data model.

\smallskip \noindent {\bf Tables.}  For a number $n \in \nat$, we denote the set
$\{1,\ldots,n\}$ by $[n]$.  By $\bot$ we denote a special
distinguished null value. For any set $\cV$, we denote the set
$\cV\cup\{\bot\}$ by $\cVbot$. The W3C formalizes tabular documents
through \emph{tables}, which can be defined as follows.

\begin{definition}[Core Tabular Data Model, \cite{w3c-tabular-data-model}]\upshape\label{def:tabular}
  Let $\cV$ be a set. A \emph{table} over $\cV$ is an $n \times m$
  matrix $T$ (for some $m,n \in \nat)$ in which each cell carries a
  value from $\cVbot$.  We say that $T$ has $n$ \emph{rows} and $m$
  \emph{columns}.     A \emph{(table) coordinate} is an element of $[n]\times[m]$.  A
  \emph{cell} is determined by coordinate $(k,\ell) \in [n] \times [m]$
  and its \emph{content} is the value $T_{k,\ell} \in \cVbot$ at the
  intersection of row $k$ and column $\ell$. We denote the
  \emph{coordinates} of $T$ by $\coords(T)$.
\end{definition}

\noindent {\bf Tabular documents.} Notice that tables are always
rectangular\footnote{Tables are required to be rectangular by Section
  2.1 of \cite{w3c-tabular-data-model}; as by paragraph 2.4 of the
  memo RFC4180 on CSV (http://tools.ietf.org/html/rfc4180).} whereas,
in Section~\ref{sec:examples}, this was not the case for some of the
use cases. We model this by padding shorter rows by
$\bot$. More precisely, we see the correspondence between
\emph{tabular documents}, i.e., text files that describe tabular data
(like CSV files), and tables as follows. Let $\alphabet$ be a finite
alphabet and let $D$ be a set of \emph{delimiters}, disjoint from
$\alphabet$. We assume that $D$ contains two designated elements which
we call \emph{row delimiter} and \emph{column delimiter}, which, as
the name indicates, separate cells vertically or
horizontally.(We discuss other delimiters in Section~\ref{sec:extensions}.)  Therefore, a sequence of symbols in
$(D\cup\alphabet)^*$ can be seen as a table over $\alphabet^*$:
every row delimiter induces a new row in the table, every column
delimiter a new column, and the $\alphabet$-strings between delimiters
define the cell contents. In the case that some rows have fewer
columns than others, missing columns are expanded to the right and
filled with $\bot$. Conversely, a table over $\alphabet^*$ can also be
seen as a string over $(D\cup\alphabet)^*$ by concatenating all its
cell values in top-down left-to-right order and inserting cell delimiters
and row delimiters in the correct places; we do not insert column
delimiters next to $\bot$-cells. As such, when we convert a tabular
document into a table and back; we obtain the original tabular
document. 

We consider both representations in the remainder of the paper. In particular
we view the table representation as a structure that allows efficient
navigation in all directions and the string representation as
structure for streaming validation. 
\smallskip\noindent {\bf \corechisel schemas.} Abstractly speaking, a
\corechisel schema $S$ is a $4$-tuple $(D, \tokens, \tokendef, R)$
where $D$ is the set of delimiters; $\tokens$ is a finite set of
tokens; $\tokendef$ is a mapping that associates a regular expression
over $\alphabet$ to each token $\tau \in \tokens$; and $R$ is a
\emph{tabular schema}, a set of rules that constrain the admissible
table contents (further defined below).

Checking whether a tabular document $\sigma$ in $(D\cup\alphabet)^*$
satisfies $S$ proceeds conceptually in three phases. In the first
phase, the delimiters are used to parse $\sigma$ into a table $T^\text{raw}$ over
$\alphabet^*$, as described above. In the second phase, the token
definitions $\tokendef$ are used to transform $T^\text{raw}$ into a \emph{tokenized}
table $T$, which is a table where each cell contains a set of tokens
(i.e., each cell contains a subset of $\Delta$, namely those tokens
that match the cell). Formally,
$T$ is the table of the same dimension as $T^\text{raw}$  such that 
\[ T_{k,\ell} = \{ \tau \in \tokens \mid T^\text{raw}_{k,\ell} \in
\lang(\tokendef(\tau)) \}. \] Here $\lang(\cdot)$ denotes the language
of a regular expression.  Finally, the rules in $R$ check validity of
the tokenized table $T$ (and not of the raw table $T^\text{raw}$), as
explained next.

\smallskip \noindent {\bf Tabular schema.}  The \emph{tabular schema}
$R$ describes the structure of the tokenized table. Intuitively, a
tabular schema is a set of rules $s \to c$ in which $s$ \emph{selects}
a region in the table and $c$ describes what the \emph{content} of the
selected region should be. More formally, a \emph{region} $z$ of an $n
\times m$ table $T$ is a subset of $[n] \times [m]$. A \emph{region
  selection language $\cS$} is a set of expressions such that every $s
\in \cS$ defines a region in every table $T$. More precisely, $s[T]$
is always a (possibly empty) region of $T$. A \emph{content language} $\cC$
is a set of expressions such that every $c\in\cC$ maps each region $z$
of $T$ to true or false. We denote by $T, z\models c$ that $c$ maps $z$
to true in $T$ and say that $z$ \emph{satisfies} $c$ in $T$.

\begin{definition}[Tabular Schema]\upshape
  A \emph{(tabular) sche\-ma} (over $\cS$ and $\cC$) is a set $R$ of
  rules $s \to c$ for which $s \in \cS$ and $c \in \cC$. A table $T$
  \emph{satisfies} $R$, denoted $T \models R$, when for every rule $s
  \to c \in R$ we have that $T, s[T]\models c$.
\end{definition}

The above definition is very general as it allows arbitrary languages
for selecting regions and defining content. We now propose concrete
languages for these purposes.

\smallskip\noindent{\bf Region selection
  expressions.}\label{sec:selector-expressions}
Our region selection language is divided into two sorts of
expressions: \emph{coordinate expressions} (ranged over by $\varphi, \psi$)
and \emph{navigational expressions} (ranged over by $\alpha$, $\beta$),
defined by the following syntax:
\begin{align*}
  \varphi, \psi & := a \mid \textsf{root} \mid \textsf{true} \mid \varphi \lor \psi \mid
  \varphi \land \psi  
  \mid \lnot \varphi \mid \langle \alpha \rangle \mid \alpha(\varphi) \\
  \alpha, \beta & := \varepsilon \mid \up \mid \down \mid \left \mid
  \right \mid [\varphi] \mid (\alpha \cdot \beta) \mid (\alpha | \beta)
  \mid (\alpha^*)
\end{align*}
Here, $a$ ranges over tokens in $\Delta$ and \textsf{root} is a
constant referring to coordinate $(1,1)$. When evaluated over an $n
\times m$ table $T$ over $2^\Delta$, a coordinate expression $\varphi$
defines a region $\sem{\varphi}_T \subseteq \coords(T)$, whereas a navigational
expression $\alpha$ defines a function $\sem{\alpha} : 2^{\coords(T)} \to
2^{\coords(T)}$, as follows.
\begin{align*}
  \sem{a}_T & := \{ (i,j) \in [n] \times [m] \mid a \in T_{i,j} \} \\
  \sem{\textsf{root}}_T & := \{(1,1)\}\\
  \sem{\textsf{true}}_T & := [n] \times [m] \\
  \sem{(\varphi \lor \psi)}_T & := \sem{\varphi}_{T} \cup \sem{\psi}_{T} \\
  \sem{(\varphi \land \psi)}_T & := \sem{\varphi}_T \cap \sem{\psi}_T \\
  \sem{(\lnot \varphi)}_T & := ([n] \times [m]) \setminus \sem{\varphi}_T \\
  \sem{\langle \alpha \rangle}_T & := \{ c \in \coords(T) \mid \sem{\alpha(\{c\})}_T \neq
  \emptyset\} \\ 
  \sem{\alpha(\varphi)}_T & := \sem{\alpha(\sem{\varphi}_T)}_T
\end{align*}
Furthermore, for every set of coordinates $C \subseteq \coords(T)$,
\begin{align*}
  \sem{\varepsilon(C)}_T & := C\\
  \sem{\up(C)}_T & := \{ (i-1,j) \mid (i,j) \in C, i > 1\} \\
  \sem{\down(C)}_T & :=\{ (i+1,j) \mid (i,j) \in C, i < m\} \\
  \sem{\left(C)}_T & :=\{ (i,j-1) \mid (i,j) \in C, j > 1\} \\
  \sem{\right(C)}_T & := \{ (i,j+1) \mid (i,j) \in C, j < n\} \\
  \sem{[\varphi](C)}_T & := C \cap \sem{\varphi}_T \\
  \sem{(\alpha \cdot \beta)(C)}_T & := \sem{\beta(\sem{\alpha(C)}_T)}_T\\
  \sem{(\alpha | \beta)(C)}_T & := \sem{\alpha(C)}_T \cup \sem{\beta(C)}_T \\
  \sem{(\alpha^*)(C)}_T & := \bigcup_{i \geq 0} \sem{\alpha^i(C)}_T 
\end{align*}
Here,  $\alpha^i(C)$ abbreviates the $i$-fold composition
$\alpha \cdots\alpha(C)$. We also use this abbreviation in the remainder.
Notice that every coordinate $(k,\ell)$ of $T$ can be expressed
as $\down^{k-1} \cdot \right^{l-1}(\textsf{root})$.
For navigational expressions $\alpha$, we abbreviate $\alpha \cdot\alpha^*$ by
$\alpha^+$ and $\alpha | \varepsilon$ by $\alpha?$. One can read
$\alpha(\varphi)$ as ``apply the regular expression $\alpha$ to
$\varphi$''. The definition of the semantics of $\alpha\cdot\beta$ is
conform with this view. 

\begin{example}\upshape
  Region selection expressions navigate in tables, similar to how XPath
  expressions navigate on trees. For example, assuming \textsf{dummy} to be
  a token for {\tt -99.00} in Figure~\ref{fig:climate}, the
  expression
  $$(\right^+(\textsf{root}) \land \lnot (\up^*(\textsf{dummy})))$$ 
selects the top cells of columns that do not contain a dummy value
anywhere.  In the excerpt of Figure~\ref{fig:climate}, this expression
hence selects the cell containing {\tt ENTEBBE AIR}.

Assuming the token \textsf{literal} for cells with quotation marks
(regex \verb|\"[a-zA-Z0-9]\"|) in Figure~\ref{ex:2:csv}, the expression
  $$ \down^+ \cdot [\textsf{literal}] \cdot \right^+ (\textsf{object}) $$
selects all provenance information for rows in which the {\tt object}
is between quotes, like {\tt "Bart", "JoJo", and "10"}. Notice in particular that the semantics of the operator
``$[\ ]$'' in navigational expressions is the same as
filter-expressions in XPath. \hfill $\blacksquare$
\end{example}

Readers familiar with propositional dynamic logic~\cite{FischerL-jcss79} (PDL for short) will
recognize that the above language is nothing more than
propositional dynamic logic, tweaked to navigate in tables.

There are some differences between the syntax of \corechisel and the
region selection expressions used in the examples of
Section~\ref{sec:chisel:examples}. In particular, the latter examples use the
following abbreviations.
\begin{remark}\upshape
  (i) As already observed, absolute coordinates in
  Section~\ref{sec:chisel:examples} are syntactic sugar for
  navigations that start at the root. For example, the coordinate
  $(2,2)$ would be unfolded to $\down \cdot \right(\textsf{root})$ in
\corechisel.

(ii) The keywords {\tt row} and {\tt col} in
Section~\ref{sec:chisel:examples} are syntactic sugar for $\right^+$
and $\down^+$ in \corechisel, respectively. So, {\tt col((2,2))},
which denotes \emph{the column below the cell $(2,2)$} in
Section~\ref{sec:chisel:examples}, is syntactic sugar for
$\down^+(\down \cdot \right(\textsf{root}))$.

(iii) The only exception to rule (ii) above are row and
column expressions of the form $\row(k)$ and $\col(\ell)$. These
abbreviate $\right^*(k,1)$ and $\down^*(1,\ell)$, respectively. (Where
$(k,1)$ and $(1,\ell)$ need to be further unfolded themselves.) 

As an example, the selection expression {\tt row(1)} of Figure
\ref{ex:2:schema} can be written as $\right^*(\textsf{root})$ or,
equivalently, $\right^*$ and the expression {\tt col(subject)} as
$\down^+(\text{\texttt{subject}})$.  \hfill $\blacksquare$

\end{remark}

\smallskip\noindent
{\bf Content expressions.}
\label{sec:content-expressions}
A \emph{content expression} is simply a regular expression $\rho$ over
the set of tokens $\Delta$.  To define when a region in a tokenized
table $T$ is valid with respect to content expression $\rho$, let us
first introduce the following order on coordinates. We say that
coordinate $(k,\ell)$ precedes coordinate $(k', \ell')$ if we visit $(k,\ell)$
earlier than $(k',\ell')$ in a left-to-right top-down traversal of the
cells of $T$, i.e., it precedes it in
lexicographic order. Formally, $(k, \ell) < (k', \ell')$ if $k < k'$ or if $k =
k'$ but $\ell < \ell'$. 

Now, let $T$ be a tokenized table, let $z$ be a region of $T$, and let
$\rho$ be a content expression. Then $(T,z)$ satisfies the content
expression $\rho$ \emph{under the region-based semantics}, denoted $T, z
\regionmodels \rho$ if there exist tokens $a_1,\dots, a_n \in \Delta$ such that
$a_1 \dots a_n \in \lang(\rho)$ and $a_i \in T_{c_i}$, where $c_1, \dots,
c_n$ is the enumeration in table order of all coordinates in $z$.

To define the row-based semantics we used in
Section~\ref{sec:examples}, we require the following notions. Let $z$
be a region of $T$. We say that subregion $z' \subseteq z$ is a
\emph{row of} $z$ if there exists some $k$ such that $z' = \{ (k,\ell)
\mid (k,\ell) \in z \}$. Now, $(T,z)$ satisfies the content expression
$\rho$ \emph{under the row-based semantics}, denoted $T,z \models
\rho$, if for every row $z'$ of $z$, we have $T,z' \regionmodels
\rho$.

\begin{remark}\upshape
  Recall that, for ease of exposition, we allowed tables to be
  non-rectangular in Section~\ref{sec:chisel:examples} whereas in our
  formal model, tables are always rectangular. In particular, shorter
  rows are padded with $\bot$ to obtain rectangularity. This implies
  that, some content expressions of Section~\ref{sec:chisel:examples}
  need to be adapted in our formal model. For example, the rule {\tt
    row(1) -> name} of Figure~\ref{ex:1:schema} needs to be adapted
  to $\row(1) \to \text{\texttt{name}}, \bot, \bot, \bot$ to 
  take the padding into account. \hfill $\blacksquare$
\end{remark}

\makeatletter{}\section{Efficient Validation}\label{sec:evaluation}

In this section we consider the  \emph{validation} (or
\emph{evaluation}) problem for
tabular schemas. This problem asks, given a tokenized table or tabular
document $T$ and a tabular schema $R$, whether $T$ satisfies $R$. We
consider the problem in a \emph{main-memory} and \emph{streaming}
variant. Intuitively, $T$ is given as a table in the former and as a
tabular document in the latter setting. 

\subsection{Validation in Linear Time}\label{sec:evaluation:linear}

When $T$ is given as a tokenized table, we can essentially assume that
we can navigate from a cell $(i,j)$ to any of its four neighbours
$\up(\{(i,j)\})$, $\down(\{(i,j)\})$, $\left(\{(i,j)\})$, and
  $\right(\{(i,j)\})$ in constant time. Under these assumptions we
show that $T$ can be validated against a tabular schema in linear time
combined complexity.\footnote{Combined complexity is a standard
  complexity measure introduced by Vardi; see~\cite{Vardi-stoc82}.}
The proof strongly relies on the linear time combined complexity for
PDL model checking.
\begin{theorem}\label{theo:evaluation}\label{THEO:EVALUATION}
  The valuation problem for a tabular document $T$ and a schema $R$
  is in linear time combined
  complexity, that is, time $O(|T||R|)$.
\end{theorem}

\makeatletter{}\subsection{Streaming Validation}
\label{sec:streaming-evaluation}

Even though Theorem~\ref{theo:evaluation} implies that 
\chisel schemas can be efficiently validated, the later claim only holds 
true when the tabular document can be fully loaded in memory and multiple passes can be made through the document. However, 
when the input data is large it is sometimes desirable to have a \emph{streaming validation algorithm}
that makes only a single pass over the input tabular document and uses
only limited memory. In this section we identify several fragments of
\corechisel that admits such streaming validation algorithms.

\newcommand{\event}[1]{\langle \textnormal{\textsf{cell }} #1\rangle}

\smallskip\noindent {\bf Streaming model.} Let us begin by defining
when an algorithm validates in a streaming fashion. In
this respect, we draw inspiration from the SAX Streaming API for XML:
we can view a tokenized table $T$ as a sequence of events generated by
visiting the cells of $T$ in table order. Here, whenever we visit a new
cell, an event $\event{\Gamma}$ is emitted, with $\Gamma$ the set of
tokens in the visited cell. Whenever we move to a new row, an event of
type \newrow\ is emitted.

Note that the tokenized event stream can easily be generated ``on the
fly'' when parsing a tabular document: we start reading the tabular
document, one character at a time, until we reach a delimiter. All
non-delimiter characters are used as input to, e.g., a finite state automaton
that allows us to check which tokens match the current cell's
content. When we reach a delimiter, a $\event{\Gamma}$ event is emitted with
the corresponding set of matching tokens. If the delimiter is a row delimiter,
then also a \newrow\ is emitted. We repeat this until the end of the file.
\begin{example}\upshape
  Consider the tabular document from Figure~\ref{fig:climate} together
  with the corresponding \chisel schema $S$ in
  Figure~\ref{fig:climate-schema}. The tokenized table of this document
  according to $S$ yields the event stream
  \begin{multline*}
    \event{\emptyset} \event{ \{\textnormal{\texttt{ARUA}}\}}  \event{
    \{\textnormal{\texttt{BOMBO}}\}} \event{
    \{\textnormal{\texttt{ENTEBBE AIR}}\}} \\ \newrow 
\event{\{\textnormal{\texttt{Timestamp}}\}}
\event{\{\textnormal{\texttt{Temperature}}\}} \\
\event{\{\textnormal{\texttt{Temperature}}\}} 
\event{\{\textnormal{\texttt{Temperature}}\}} \newrow \dots
  \end{multline*}
\end{example}

\begin{definition}[Streamability]\upshape
A tabular schema $R$ is said to be \emph{weakly streamable}, if there
exists a Turing Machine $M$ that
\begin{compactitem}[-]
\item can only read its input tape once,  from left to right;
\item for every tokenized table $T$, when started with the event
  stream of $T$ on its input tape, accepts iff $T \models R$; and
\item has an auxiliary work tape that can be used during processing,
  but it cannot use more than $O(m \log(n))$ of space on this work
  tape, where $n$ is the total number of cells in $T$, and $m$ the number of
  columns.
\end{compactitem}
We say that $R$ is \emph{strongly streamable} if the Turing Machine $M$
only requires $O(\log(n))$ space on its work tape.
\end{definition}
Here, strong streamability corresponds to the commonly
studied notions of streaming evaluation. We consider weak
streamability to be very relevant as well because, based on the W3C use
cases, tabular data often seems to be similar in spirit to relational
tables and, in these cases, is very narrow and deep. In particular, $m
= O(\log n)$ in these cases.

\smallskip\noindent {\bf Weak streamability.} To enable streaming
validation, we restrict our attention to so-called \emph{forward}
coordinate and navigational expressions which are expressions where
$\langle \alpha \rangle$ is not allowed, and we never look $\up$ or
$\left$. That is, a coordinate or navigational expression is
  \emph{forward} if it is generated by the following syntax.
\begin{align*}
  \varphi, \psi & := a \mid \textsf{root} \mid \textsf{true} \mid \varphi \lor \psi \mid
  \varphi \land \psi  
  \mid \lnot \varphi \mid  \alpha(\varphi) \\
  \alpha, \beta & := \varepsilon \mid \down \mid 
  \right \mid [\varphi] \mid (\alpha \cdot \beta) \mid (\alpha + \beta)
  \mid (\alpha^*)
\end{align*}
We do not consider the operator $\langle \alpha \rangle$ in the
forward fragment because it can be seen as a backward operator:
$\langle \right \cdot [a] \rangle$ is equivalent to $\left(a)$. 

A \corechisel schema is \emph{forward} if it mentions only forward
coordinate expressions.

\begin{theorem}\label{theo:weak-stream}\label{THEO:WEAK-STREAM}
  Forward \corechisel is weakly streamable.
\end{theorem}
\begin{proof}[sketch]
     Consider a rule $\varphi \to \rho$ with
  $\varphi$ a forward coordinate expression and $\rho$ a content
  expression. We will show that  coordinate
  expressions $\varphi$ can be evaluated in a streaming fashion by constructing a  special kind of finite state automaton (called \emph{coordinate automaton}) that allows us to decide, at each position in the
  event stream, if the currently visited cell is in
  $\sem{\varphi}_{T}$. Whenever we find that this is the case, we
  apply the current cell contents to $\rho$ (which we also evaluate by
  means of a finite state automaton). Now observe that $T\models
  \varphi \to \rho$ iff (1) under the row based semantics, whenever we
  see $\newrow$, the automaton for $\rho$ is in a final state and (2)
  under the region-based semantics, when we reach the end of the event
  stream, the automaton for $\rho$ is in a final state. We then obtain
  weak streamability by showing that coordinate automata for $\varphi$
  can be simulated in space $O(m \log(n))$, whereas it is known that
  the finite state automaton for $\rho$ can be simulated in constant
  space.
\end{proof}
\smallskip\noindent {\bf Strong streamability.} Forward \corechisel is
not strongly streamable as no schema with a rule that
contains subexpressions of the form $\col(a)$ (which are prevalent in
Section~\ref{sec:examples}) can be strongly streamable. This can be
seen using a simple argument from communication complexity. Indeed, assume
that the first row has $k$ cells, some of which have the token $a$ and
some of which do not. If we want to evaluate $\col(a)$ in a streaming
fashion, we need to identify the cells in the second row that are in
the same columns as the $a$-tokens in the first row.  But, this is
precisely the equality of two $k$-bit strings problem, which
requires $\Omega(k)$ bits in deterministic communication complexity
(Example 1.21 in \cite{KushilevitzN}). These $\Omega(k)$ bits are what
we need to store when going from the first to the second row. Since
$k$ can be $\Theta(n)$, this amount of space is more than we allow for
strongly streamable tabular schemas, and hence no schema containing
$\col(a)$ is strongly streamable.

The underlying reason why $\col(a)$ is not strongly streamable is because, in
general, the token $a$ can occur arbitrarily often. However, in all such
cases in Section~\ref{sec:examples} and in the W3C use cases, the occurrences of $a$ are
very restricted. We could therefore obtain strong
streamability for such expressions by adding constructs in the
language that restrict how certain tokens can appear:

\centerline{\texttt{unique(a)} \qquad \qquad \texttt{unique-per-row(a)}}

\noindent The former asserts that token $a$ should occur only once in the
whole table and the latter that $a$ occurs at most
once in each row. More formally, the former predicate
holds in a table $T$ if $\sem{a}_T$ contains at most one element and
the latter holds in table $T$ if $\sem{a}_T$ contains at most one
element of the form $(r,c)$ for each row number $r$. Notice that a
strong streaming algorithm can easily check whether these
predicates hold.

We use the above predicates to define two notions of
\emph{guardedness} for region selection expressions. Guarded formulas
will be strongly streamable. We say that token $a$ is
\emph{row-guarded} if {\tt unique-per-row(a)} appears in the
schema. If {\tt unique(a)} appears in the schema it is, in addition,
also \emph{guarded}. The two notions of guardedness capture
the following intuition: if $\varphi$ is row-guarded, then 
$\down(\varphi)$ is strongly streamable and if it is guarded, then
$\down^*(\varphi)$ is strongly streamable. The main idea is that, in
both cases, the number of cells we need to remember when going from
one row to the next does not depend on the width of the table.
We now
define (row)-guardedness inductively on the forward language:
\begin{compactitem}
\item $\textsf{root}$ and $\true$ are guarded and row-guarded;
\item $\right^*(\varphi)$ is guarded and row-guarded for every
$\varphi$ that does not contain a navigational subexpression;
\item if $\varphi$, $\psi$, $\alpha(\varphi)$, $\beta(\varphi)$ are
  guarded (resp., row-guarded), then 
  \begin{compactitem}
  \item $\varphi \land \psi$,  $\varphi \lor \psi$, $\varepsilon(\varphi)$, $\down(\varphi)$, $\right(\varphi)$, 
  \item $(\alpha\cdot\down)(\varphi)$, $(\alpha\cdot\right)(\varphi)$, and $(\alpha + \beta)(\varphi)$
\end{compactitem}
are guarded (resp.,  row-guarded);
\item if $\varphi$ and $\alpha(\psi)$ are guarded then
  $\down^*(\varphi)$ and  $(\alpha\cdot\down^*)(\psi)$ are guarded; and
\item if $\varphi$ and $\alpha(\psi)$ are row-guarded then
  $\right^*(\varphi)$ and $(\alpha\cdot \right^*)(\psi)$ are guarded.
\end{compactitem}

\begin{definition}\upshape
  A forward core-\chisel schema is \emph{guarded}, if all region
  selection expressions that use the $\down$-operator are row-guarded
  and all region selection expressions that use $\down^*$ are
  guarded.
\end{definition}
Notice that guardedness of a \chisel schema can be tested in linear
time. Furthermore notice that every \chisel schemas in this paper
becomes strongly streamable if we add the predicates {\tt unique(a)}
for tokens $a$ that we use in expressions using $\col$, $\down$, or
$\down^*$.

\begin{theorem}
  \label{thm:forward-guarded-strong-streaming}
  \label{THM:FORWARD-GUARDED-STRONG-STREAMING}
  Guarded forward core-\chisel is strongly streamable.
\end{theorem}

\makeatletter{}\section{\tdsltitle\ Extensions}\label{sec:extensions}

Next, we describe a number of extensions to \tdsl. These include 
alternative grouping semantics, types, complex content cells, and a concept for a transformation language.

\subsection{Region semantics}
\makeatletter{}
The examples in Section~\ref{sec:examples} all use a row-based
semantics of \tdsl where the content expression is matched over every
row in the selected region. That is, the
cells in the selected region are `grouped by' the row they occur
in. There are of course other ways to group cells, by
column, for instance, or by not grouping them at all.
The latter case is already defined in Section~\ref{sec:formalmodel} as
\emph{region-based semantics}.  In \tdsl, we indicate rules using this semantics with a double arrow {\tt =>} rather than a single arrow.  
Notice the difference
between the rules {\tt col(2) -> Null | Number} and {\tt col(2) =>
  (Null | Number)*}. Both require each cell in the second column to be
empty or a number but express this differently. (The former way is
closer to how one defines the schema of a table in SQL, which is why
we chose it as a default.)
Example~\ref{ex:PDB}
describes a more realistic application of {\tt =>}-rules. This example
corresponds to use case 12 in \cite{w3c-csv-usecases}, is called
\emph{``Chemical Structures''} and aims to interpret Protein Data Bank
(PDB) files as tabular data. This particular use case is interesting
because it illustrates that the view of W3C on tabular data is not
restricted to traditional comma-separated values files.
We note that Theorems~\ref{theo:evaluation}, \ref{theo:weak-stream}, and \ref{thm:forward-guarded-strong-streaming} still hold if \tdsl
schemas contain both rules under row-based and region-based semantics.

\begin{figure*}
\begin{framed}
  \vspace{-2mm}
\begin{verbatim}
HEADER    EXTRACELLULAR MATRIX                    22-JAN-98   1A3I
TITLE     X-RAY CRYSTALLOGRAPHIC DETERMINATION OF A COLLAGEN-LIKE
TITLE    2 PEPTIDE WITH THE REPEATING SEQUENCE (PRO-PRO-GLY)
...
EXPDTA    X-RAY DIFFRACTION
AUTHOR    R.Z.KRAMER,L.VITAGLIANO,J.BELLA,R.BERISIO,L.MAZZARELLA,
AUTHOR   2 B.BRODSKY,A.ZAGARI,H.M.BERMAN
...
REMARK 350 BIOMOLECULE: 1
REMARK 350 APPLY THE FOLLOWING TO CHAINS: A, B, C
REMARK 350   BIOMT1   1  1.000000  0.000000  0.000000        0.00000
REMARK 350   BIOMT2   1  0.000000  1.000000  0.000000        0.00000
...
SEQRES   1 A    9  PRO PRO GLY PRO PRO GLY PRO PRO GLY
SEQRES   1 B    6  PRO PRO GLY PRO PRO GLY
SEQRES   1 C    6  PRO PRO GLY PRO PRO GLY
...
ATOM      1  N   PRO A   1       8.316  21.206  21.530  1.00 17.44           N
ATOM      2  CA  PRO A   1       7.608  20.729  20.336  1.00 17.44           C
ATOM      3  C   PRO A   1       8.487  20.707  19.092  1.00 17.44           C
ATOM      4  O   PRO A   1       9.466  21.457  19.005  1.00 17.44           O
ATOM      5  CB  PRO A   1       6.460  21.723  20.211  1.00 22.26           C
\end{verbatim}
\vspace{-5mm}
\end{framed}
\caption{Fragment of a PDB file.\label{fig:PDB}}
\end{figure*}

\begin{example}\label{ex:PDB}\upshape
  Figure \ref{fig:PDB} displays a slightly shortened version of the
  PDB file mentioned in use case 12 in \cite{w3c-csv-usecases}.
The corresponding \tdsl schema could contain the following rules:
\vspace{-2mm}
\begin{verbatim}
row(1) -> HEADER, Type, Date, ID
col(1) => HEADER, TITLE*, dots, EXPDATA, AUTHOR*, 
          dots, REMARK*, dots, SEQRES*, dots, ATOM*          
\end{verbatim}
\vspace{-2mm}
 The last rule employs the region semantics and 
specifies the order in which tokens in the first column should appear.
\hfill $\blacksquare$

\end{example}

\subsection{Token types}

  The PDB fragment in Figure \ref{fig:PDB} contains cells that have
  the same content but seem to have a different meaning.   It can be convenient to differentiate between cells by using
  \emph{token types} as follows:
  \vspace{-1mm}
\begin{verbatim}
%% Token types
%% left: name of the token type
%% right: region selection expression for token type
REMARK-Header <= down*[dots].down[REMARK]
REMARK-Comment <= down*[dots].down[REMARK].down
REMARK-Rest <= 
  down*[dots].down[REMARK].down.(down[REMARK])*
\end{verbatim}
  \vspace{-1mm}
  Note that we abbreviated rules of the form $\alpha(\textsf{root})$
  by $\alpha$. We denoted the concatenation operator of navigational
  expressions by ``{\tt .}''.
      \texttt{REMARK-Header} is the topmost cell containing \texttt{REMARK} in
  Figure \ref{fig:PDB}, \texttt{REMARK-Comment} is the one immediately below, and
  \texttt{REMARK-Rest} is the rest.  We can now use token types to write
  rules such as
  \vspace{-1mm}
\begin{verbatim}
row(REMARK-Header) -> ...
row(REMARK-Comment) -> ...
row(REMARK-Rest) -> ...
\end{verbatim}
  \vspace{-1mm}
  Token types do not add additional expressiveness to the language
  since one can simply replace \texttt{REMARK-HEADER} by
  \texttt{down*[dots].down[REMARK](root)} in the rule. But the ability to
  use different names for fields with the same content may be useful
  for writing more readable schemas. In this case, the names suggest
  that the block of remarks is divided into a header, some comment,
  and the rest.

\subsection{Transformations and Annotations}

While it is beyond the scope of this document to develop a
transformation language for tables, we argue that region selection
expressions can be easily employed as basic building blocks for a
transformation language aimed at transforming tables into a variety of
formats like, for instance, RDF, JSON, or XML (one of the scopes
expressed in \cite{w3cchartercsv}). Region selection
expressions are then used to identify relevant parts of a
table.

\noindent {\bf Basic Transformations.}  Consider Figure~\ref{fig:climate} (of
Example~\ref{ex:simple}) again, where we see that several columns have
the value $-99.00$. Since winter does not get this extreme in Uganda,
this value is simply a dummy which should not be considered when
computing, e.g., the average temperature in Uganda in 1935. Instead,
for the fragment of Figure~\ref{fig:climate}, it would be desirable to
only select the columns that do not contain $-99.00$. To do this,
we can simply define a new token and a new token type for the region
of the table we are interested in.
  \vspace{-1mm}
\begin{verbatim}
Useless-Temp = -99.00
  
%% Token type 
Useful <= col(1) or 
         (Temperature and not Useless-Temp) or 
         (row(1) and not up*(Useless-Temp))
\end{verbatim}
  \vspace{-1mm}
 The region defined by {\tt Useful} contains
  \vspace{-1mm}
\begin{verbatim}
       , ENTEBBE AIR
1935.04,       27.83
1935.12,       25.72
1935.21,       26.44
[...]
\end{verbatim}
  \vspace{-1mm}
 which could then be exported. Using simple for-loops we can iterate over rows, columns, or cells,
and compute aggregates. For example, 
  \vspace{-1mm}
\begin{verbatim}
Useful-values <= (Temperature and not Useless-Temp)

For each column c in Useful-values {
  print Average(c)
}
\end{verbatim}
  \vspace{-1mm}
 would output {\tt 25.65}, the average of the values below {\tt ENTEBBE
  AIR} in Figure~\ref{fig:climate}. The region defined by {\tt
  Useful-values} is a set of table cells, with coordinates. These
coordinates can be used to handle information column-wise in the
for-loop: It simply iterates over all column coordinates that are
present in the region. Iteration over rows or single cells would work analogously.

\noindent {\bf Namespaces, Annotations and RDF.}
Assume that we want to say that certain cells in Figure~\ref{ex:1:csv}
are geographical regions. To this end, the \tdsl schema could contain a
definition of a default namespace:

\noindent \texttt{namespace default = http://foo.org/nationalstats.csv}

\noindent \texttt{namespace x = [...]}

\noindent  Region selection expressions can then be used to specify which cells
should be treated as objects in which namespace. For example, the code fragment
\vspace{-1mm}
\begin{verbatim}
For each cell c in col(GeoArea) {
  c.namespace = default
}
\end{verbatim}
\vspace{-1mm}
 could express that each cell below {\tt GeoArea} is an entity in namespace
{\tt http://foo.org/nationalstats.csv}. So, the cell containing
{\tt England} represents the entity

\centerline{{\tt
  http://foo.org/nationalstats.csv:England,}}

\noindent similar for the cell
containing {\tt Wales}, etc. (Here we assume that {\tt .namespace} is
a predefined operation on cells.)

We can also annotate cells with meta-information (as is
currently being considered in Section 2.2 of \cite{w3c-tabular-data-model}). The code fragment
\vspace{-1mm}
\begin{verbatim}
For each cell c in col(GeoArea) {
  annotate c with "rdf:type dbpedia-owl:Place"
  annotate c with "owl:sameAs fbase:" + c.content
}
\end{verbatim}
\vspace{-1mm}
 (assuming appropriate namespace definitions for rdf, owl, etc.)  could
express that each cell below {\tt GeoArea} should be annotated with
{\tt rdf:type dbpedia-owl:Place} and, in addition, the 
{\tt England} cell with {\tt owl:sameAs fbase:England}, the 
{\tt Wales} cell with {\tt owl:sameAs fbase:Wales}, etc. We assume that
{\tt annotate}, {\tt with}, and {\tt .content} are reserved
words or operators in the language.

These ingredients also seem useful for exporting to RDF. We could
write, e.g.,

\noindent \texttt{print "@prefix : <http://foo.org/nationalstats.csv>"}

\begin{verbatim}
For each cell c in col(GeoArea) {
 print ":"+c.content+"owl:sameAs fbase:"+c.content
}
\end{verbatim}
\vspace{-1mm}
 to produce an RDF file that says that {\tt :England} in the default
namespace is the same as {\tt fbase:England}. 
Looking at Figure~\ref{ex:2:csv}, one can also imagine constructs like
\vspace{-1mm}
\begin{verbatim}
RDF <= col(subject) or col(predicate) or col(object)

For each row r in RDF {
  print r.cells[1] +" "+ r.cells[2] +" "+ r.cells[3]  
}
\end{verbatim}
\vspace{-1mm}
 to facilitate the construction of RDF triples taking content from
several cells.

\subsection{Complex content}

The CSV on the Web WG is considering allowing complex content (such as lists) in
cells (Section 3.8 in \cite{w3c-metadata-vocabulary}).  \tdsl can be
easily extended to reason about complex content. Our formal definition of
tabular documents already considers (Section~\ref{sec:formalmodel}) a
finite set of delimiters, which goes beyond the two delimiters (row-
and column-) that we used until now.

In a spirit similar to region-based semantics, one can also imagine a
subcell-based semantics, for example, a rule of the form
\vspace{-1mm}
\begin{verbatim}
col(1) .> (String)*
\end{verbatim}
\vspace{-1mm}
 could express that each cell in the first column contains a list of
Strings. Notice the use of {\tt .>} instead of {\tt ->} to denote that
we specify the contents of each individual cell in the region, instead
of each row. The statement {\tt List Delim = ;} in the beginning of
the schema could say that the semicolon is the delimiter for lists
within a cell.

\makeatletter{}\section{Conclusions}

We presented the schema language \tdsl for tabular data on the Web
and showcased its flexibility and usability through a wide range of
examples and use cases. While region selection expressions
are at the very center of \tdsl, we think they can be more broadly 
applied.  Region selection expressions can be used, for instance, as a cornerstone for annotation- and transformation languages for tabular 
data and thus for a principled approach for integrating such data into the Semantic Web. The whole approach of \tdsl is strongly rooted in theoretical foundations and, at the same time, in well established technology such as XPath. For these reasons, we expect the language to be very robust
and, at the same time, highly accessible for users. Two prominent directions for future work are the following: (1) expand the usefulness of \tdsl by further exploring the extensions in Section~\ref{sec:extensions}; and, (2) study static analysis problems related to \tdsl and region selector expressions leveraging on the diverse box of tools from formal language theory
and logic.

\section*{Acknowledgments}
We thank Marcelo Arenas for bringing \cite{w3c-tabular-data-model} to our attention.

\bibliographystyle{abbrv}

\begin{thebibliography}{10}

\bibitem{csv-schema-national-archives}
R.~W. Adam~Retter, David~Underdown.
\newblock Csv schema 1.0: A language for defining and validating csv data.
\newblock http://digital-preservation.github.io/csv-schema/csv-\\schema-1.0.html.

\bibitem{AI00}
N.~Alechina and N.~Immerman.
\newblock Reachability logic: An efficient fragment of transitive closure
  logic.
\newblock {\em Logic Journal of the IGPL}, 8(3):325--337, 2000.

\bibitem{DBLP:conf/www/ArenasCP12}
M.~Arenas, S.~Conca, and J.~P{\'{e}}rez.
\newblock Counting beyond a yottabyte, or how {SPARQL} 1.1 property paths will
  prevent adoption of the standard.
\newblock In {\em International World Wide Web Conference (WWW)}, pages
  629--638, 2012.

\bibitem{xpath2}
A.~Berglund, S.~Boag, D.~Chamberlin, M.~F. Fern\'andez, M.~Kay, J.~Robie, and
  J.~Sim\'eon.
\newblock {XML} {P}ath {L}anguage ({XP}ath) 2.0.
\newblock Technical report, World Wide Web Consortium, January 2007.
\newblock W3C Recommendation, http://www.w3.org/TR/2007/REC-xpath20-20070123/.

\bibitem{BexGNV-www08}
G.~J. Bex, W.~Gelade, F.~Neven, and S.~Vansummeren.
\newblock Learning deterministic regular expressions for the inference of
  schemas from {XML} data.
\newblock In {\em International World Wide Web Conference (WWW)}, pages
  825--834, 2008.

\bibitem{DBLP:conf/www/BexMNS05}
G.~J. Bex, W.~Martens, F.~Neven, and T.~Schwentick.
\newblock Expressiveness of {XSD}s: from practice to theory, there and back
  again.
\newblock In {\em International World Wide Web Conference (WWW)}, pages
  712--721, 2005.

\bibitem{CleavelandS93}
R.~Cleaveland and B.~Steffen.
\newblock A linear-time model-checking algorithm for the alternation-free modal
  mu-calculus.
\newblock {\em Formal Methods in System Design}, 2(2):121--147, 1993.

\bibitem{r2rml-spec}
S.~Das, S.~Sundara, and R.~Cyganiak.
\newblock {R2RML}: {RDB} to {RDF} mapping language.
\newblock \url{http://www.w3.org/TR/r2rml/}.
\newblock W3C Recommendation 27 September 2012.

\bibitem{xsd-0}
D.~Fallside and P.~Walmsley.
\newblock {XML} {S}chema {P}art 0: {P}rimer (second edition).
\newblock Technical report, World Wide Web Consortium, October 2004.
\newblock http://www.w3.org/TR/2004/REC-xmlschema-0-20041028/.

\bibitem{FischerL-jcss79}
M.~J. Fischer and R.~E. Ladner.
\newblock Propositional dynamic logic of regular programs.
\newblock {\em J. Comput. Syst. Sci.}, 18(2):194--211, 1979.

\bibitem{mastering-regex}
J.~E.~F. Friedl.
\newblock {\em Mastering Regular Expressions}.
\newblock O'Reilly Media, 3rd edition edition, 2006.

\bibitem{DBLP:journals/jcss/GeladeN11}
W.~Gelade and F.~Neven.
\newblock Succinctness of pattern-based schema languages for {XML}.
\newblock {\em J. Comput. Syst. Sci.}, 77(3):505--519, 2011.

\bibitem{google-dspl}
Google.
\newblock {DSPL}: Dataset publishing language.
\newblock \url{https://developers.google.com/public-data/}.
\newblock Last accessed 04/11/2014.

\bibitem{DBLP:conf/www/KumarMV07}
V.~Kumar, P.~Madhusudan, and M.~Viswanathan.
\newblock Visibly pushdown automata for streaming {XML}.
\newblock In {\em International World Wide Web Conference (WWW)}, pages
  1053--1062, 2007.

\bibitem{KushilevitzN}
E.~Kushilevitz and N.~Nisan.
\newblock {\em Communication Complexity}.
\newblock Cambridge University Press, 1997.

\bibitem{tabular-data-package}
O.~K.~F. Labs.
\newblock Tabular data package.
\newblock \url{http://dataprotocols.org/tabular-data-package/}.
\newblock Version 1.0-beta-2. Last accessed 04/11/2014.

\bibitem{LibkinMV-icdt13}
L.~Libkin, W.~Martens, and D.~Vrgoc.
\newblock Querying graph databases with {XP}ath.
\newblock In {\em International Conference on Database Theory (ICDT)}, pages
  129--140, 2013.

\bibitem{DBLP:conf/pods/LosemannM12}
K.~Losemann and W.~Martens.
\newblock The complexity of evaluating path expressions in {SPARQL}.
\newblock In {\em International Symposium on Principles of Database Systems
  (PODS)}, pages 101--112, 2012.

\bibitem{DBLP:journals/pvldb/MartensNNS12}
W.~Martens, F.~Neven, M.~Niewerth, and T.~Schwentick.
\newblock Developing and analyzing xsds through bonxai.
\newblock {\em PVLDB}, 5(12):1994--1997, 2012.

\bibitem{MartensNSB-tods06}
W.~Martens, F.~Neven, T.~Schwentick, and G.~Bex.
\newblock Expressiveness and complexity of {XML} {S}chema.
\newblock {\em ACM Transactions on Database Systems}, 31(3):770--813, 2006.

\bibitem{DBLP:journals/tods/PerezAG09}
J.~P{\'{e}}rez, M.~Arenas, and C.~Gutierrez.
\newblock Semantics and complexity of {SPARQL}.
\newblock {\em {ACM} Trans. Database Syst.}, 34(3), 2009.

\bibitem{w3c-metadata-vocabulary}
R.~Pollock and J.~Tennison.
\newblock Metadata vocabulary for tabular data.
\newblock Technical report, World Wide Web Consortium (W3C), July 2014.
\newblock www.w3.org/TR/2014/WD-tabular-metadata-20140710/.

\bibitem{shape-expr}
E.~Prud'hommeaux, J.~E.~L. Gayo, and H.~Solbrig.
\newblock Shape expressions: An {RDF} validation and transformation language.
\newblock In {\em International Conference on Semantic Systems}, 2014.

\bibitem{DBLP:conf/www/RymanHS13}
A.~G. Ryman, A.~L. Hors, and S.~Speicher.
\newblock {OSLC} resource shape: {A} language for defining constraints on
  linked data.
\newblock In {\em WWW Workshop on Linked Data on the Web}, 2013.

\bibitem{SegoufinS-icdt07}
L.~Segoufin and C.~Sirangelo.
\newblock Constant-memory validation of streaming {XML} documents against
  {DTD}s.
\newblock In {\em International Conference on Database Theory (ICDT)}, pages
  299--313, 2007.

\bibitem{SegoufinV-pods02}
L.~Segoufin and V.~Vianu.
\newblock Validating streaming {XML} documents.
\newblock In {\em International Symposium on Principles of Database Systems
  (PODS)}, pages 53--64, 2002.

\bibitem{w3c-csv-usecases}
J.~Tandy, D.~Ceolin, and E.~Stephan.
\newblock {CSV} on the {Web}: Use cases and requirements.
\newblock Technical report, World Wide Web Consortium (W3C), October 2014.
\newblock http://w3c.github.io/csvw/use-cases-and-requirements/.

\bibitem{Jeni90}
J.~Tennison.
\newblock 2014: The year of {CSV}.
\newblock \url{http://theodi.org/blog/2014-the-year-of-csv}.
\newblock last accessed 04/11/2014.

\bibitem{w3c-tabular-data-model}
J.~Tennison and G.~Kellogg.
\newblock Model for tabular data and metadata on the web.
\newblock Technical report, World Wide Web Consortium (W3C), July 2014.
\newblock www.w3.org/TR/2014/WD-tabular-data-model-20140710/.

\bibitem{Vardi-stoc82}
M.~Y. Vardi.
\newblock The complexity of relational query languages (extended abstract).
\newblock In {\em {ACM} Symposium on Theory of Computing (STOC)}, pages
  137--146, 1982.

\bibitem{w3cchartercsv}
W3C.
\newblock {CSV} on the web working group charter.
\newblock http://www.w3.org/2013/05/lcsv-charter.html.

\end{thebibliography}

\newpage

\appendix
\makeatletter{}
\section{Proof of Theorem~\ref{theo:evaluation}}
\label{sec:proof-theorem-evaluation}

\begin{lemma}[See Fact 5.1, \cite{LibkinMV-icdt13}]\label{lem:PDL}
  For every coordinate expression $\varphi$, we can compute
  $\sem{\varphi}_T$ in time $O(|T||\varphi|)$. For every set $C$ of
  coordinates in $T$ and every navigational expression
  $\alpha$, we can compute $\sem{\alpha(C)}_T$ in time $O(|T||\alpha|)$.
\end{lemma}

\begin{proof}[of Theorem~\ref{theo:evaluation} (sketch)]
  As we have already mentioned, our region expressions are essentially
  a variant of Propositional Dynamic Logic (PDL), tweaked to navigate
  in tables. It is known that PDL has linear time combined complexity
  for global model checking \cite{AI00,CleavelandS93}. That is, given
  a PDL formula $\varphi$ and a Kripke structure (essentially, a
  graph) $G$, one can decide in time $O(|\varphi||G|)$ whether $G
  \models \varphi$. One can simply view our tables as Kripke
  structures that are shaped like a grid. The result follows by a
  straightforward adaptation of the algorithm in \cite{AI00}.
\end{proof}

\makeatletter{}\newcommand{\act}{\textnormal{\textit{Act}}}
\newcommand{\susp}{\textnormal{\textit{Susp}}}

\section{Proof of Theorem~\ref{theo:weak-stream}}
\label{sec:proof-theorem}

The crux behind Theorem~\ref{theo:weak-stream} is the
following. Consider a rule $\varphi \to \rho$ with $\varphi$ a forward
coordinate expression and $\rho$ a content expression. We will show that
we are able to evaluate coordinate expressions $\varphi$ in a
streaming fashion by constructing a special kind of finite state
automaton (called \emph{coordinate automaton}) that allows us to
decide, at each position in the event stream, if the currently visited
cell is in $\sem{\varphi}_{T}$. Whenever we find that this is the
case, we apply the current cell contents to $\rho$ (which we also
evaluate by means of a finite state automaton). Now observe that
$T\models \varphi \to \rho$ iff
\begin{enumerate}
\item Under the row based semantics, the automaton for $\rho$ is in a
  final state whenever we see $\newrow$.
\item Under the region-based semantics, when we reach the end of the
  event stream, the automaton for $\rho$ is in a final state.
\end{enumerate}

\subsection{Coordinate automata}
\label{sec:coordinate-automaton}

Let us first introduce the kinds of finite state automata that we will
use to evaluate coordinate expressions. A \emph{coordinate automaton}
(CA for short) $A$ over $\Lambda$ is a tuple $(Q, q_0, F, \delta)$
where:
\begin{itemize}
\item $Q$ is a finite set of states;
\item $q_0 \in Q$ is an initial state;
\item $F \subseteq Q$ is a set of \emph{final states};
\item $\delta$ is a finite transition relation consisting of triples,
  each having one of the forms $(q, \varepsilon, q')$, $(q, [\varphi],
  q')$, $(q, \right, q')$, $(q, \down, q')$, $(q, \newrow, q')$, with
$\varphi$ a \emph{forward coordinate expression}.
\end{itemize}
A CA processes event streams of tokenized tables. It largely behaves
like a normal non-deterministice finite state automaton, but has a
number of special features. First, it is equipped with a register that
stores the coordinate of the current cell being processed. Second, in
each step during execution, the automaton can either be in an
\emph{active} state, or in a \emph{suspended state}.  Only active
states can fire new transitions; suspended states remain dormant until
they become active again. Syntactically, a suspended state is simply a
pair of the form $q \colon \ell$ with $q \in Q$ and $\ell \in \nat$, $\ell \geq
1$. The semantics of a suspended state $q\colon \ell$ is that the
automaton keeps reading events from the input stream (updating the
coordinate register, but remaining in the suspended state) until the
next time that we visit a cell in column $\ell$, when the state becomes
active again. This way, the coordinate automaton can simulate moving
downwards in the table. In particular, when following a transition of
the form $(q, \down, q')$ from state $q$, the automaton records the
column number of the current position (say, $\ell$); moves to state $q'$;
but immediately suspends $q'$ until the following cell in column $\ell$
is visited. Third and finally, a CA can check whether the current cell
is selected by a forward coordinate expression $\varphi$ by means of a transition of the
form $(q, [\varphi], q')$. The CA moves from $q$ to $q'$ if the
corresponding cell in the table is selected by $\varphi$. (Note that
the CA can hence use coordinate expressions as oracles.) In contrast
to standard automata, these transitions do not cause the CA to move to
the next event, however. Moving to the next event is done by either
transitions of the form $(q, \right, q')$ or $(q, \newrow,
q')$. 

\smallskip \noindent {\bf Note.} Coordinate automata, as defined above, are very
general and powerful. Obviously, the fact that a coordinate automaton
can use coordinate expressions as oracles makes them very
powerful. Indeed, as we will see below, they are almost trivially able
to express the semantics of coordinate expressions. In order to obtain
Theorem~\ref{theo:weak-stream}, however, we will compile coordinate
expressions into CA whose oracles are restricted.

\smallskip\noindent
{\bf Semantics.} Formally, a \emph{configuration} of $A$ is a tuple
$(i,k,\ell,\varsigma)$ where $i$ is a pointer to the current event being
read from the input tape; $(k,\ell)$ is a coordinate; and $\varsigma$ is
either an element of $Q$ (an active state), or a pair of the form
$q\colon j$ with $q \in Q$ and $j \in \nat$, $j \geq 1$ (a suspended
state).

Let $T$ be a table and let $s = \sigma_1 \cdots \sigma_n$ be an event
stream for $T$, each $\sigma_i$ being either $\event{\Gamma}$ with
$\Gamma$ a set of tokens, or $\newrow$. Let $c_0 = (i,k,\ell,\varsigma)$ be a
configuration, with $1 \leq i \leq n$ and $(k,\ell)$ the coordinate of
the cell visited in $T$ when $\sigma_i$ is emitted. A \emph{run}
$\rho$ of $A$ on $s$ starting at $c_0$ is a sequence $c_0,\dots, c_m$
of configurations such that, for all $j$ with $1 \leq j < m$ one of
the following holds for $c_j = (i_j,k_j,\ell_j,\varsigma_j)$ and $c_{j+1}
= (i_{j+1}, k_{j+1},\ell_{j+1},\varsigma_{j+1})$:
\begin{itemize}
\item Transition from active state: $\varsigma_j \in Q$
\begin{enumerate}
\item $(\varsigma_j, \epsilon, \varsigma_{j+1}) \in \delta$,
  $i_{j+1} = i_{j}$, $k_{j+1}=k_{j}$,
  and $\ell_{j+1}=\ell_{j}$ (epsilon transition)
\item $(\varsigma_j, [\varphi], \varsigma_{j+1}) \in \delta$,
  $(k_j,\ell_j) \in \sem{\varphi}_T$, $i_{j+1} = i_{j}$, $k_{j+1}=k_{j}$,
  and $\ell_{j+1}=\ell_{j}$ (ordinary transition)
\item $(\varsigma_j, \right,
  \varsigma_{j+1}) \in \delta$, $\sigma_{{i_j}+1} \not = \newrow$, $i_{j+1} = i_{j}+1$,
  $k_{j+1}=k_{j}$, and $\ell_{j+1}=\ell_{j}+1$ (move right transition)
\item $(\varsigma_j, \newrow, \varsigma_{j+1}) \in \delta$,
  $\sigma_{i_j+1} = \newrow$, $i_{j+1} =
  i_{j}+2$, $k_{j+1}=k_j+1$, and $\ell_{j+1}=1$ (new row transition)
\item $(\varsigma_j, \down,
  q) \in \Delta$, $i_{j+1} = i_{j}$, $k_{j+1}=k_{j}$, $\ell_{j+1}=\ell_{j}$,
  and $\varsigma_{j+1} = q\colon \ell_j$,  (down transition)
\end{enumerate}
\item Transition from suspended state:  $\varsigma_j = q\colon
  \ell$ for some $q \in Q$ and $\ell \in \nat$ with $\ell \geq 1$.
  \begin{enumerate}
  \item $\sigma_{i_j+1} \not = \newrow$, $i_{j+1} = i_{j}+1$,
    $k_{j+1}=k_{j}$, $\ell_{j+1}=\ell_{j}+1$, and either (a) $\ell \not =
    \ell_{j+1}$ and $\varsigma_{j+1} = \varsigma_j$ (suspended right
    transition) or (b) $\ell = \ell_{j+1}$ and $\varsigma_{j+1} = q$ (wake
    up).
  \item $\sigma_{i_j+1} = \newrow$, $i_{j+1} = i_{j}+2$,
    $k_{j+1}=k_j+1$, $\ell_{j+1}=1$, and either (a) $\ell \not = 1$ and
    $\varsigma_{j+1} = \varsigma_j$ (suspended new row transition) or (b) $\ell
    = 1$ and $\varsigma_{j+1} = q$ (wake up).
\end{enumerate}
\end{itemize}

\begin{definition}
  Let $T$ be a tokenized table. A CA $A$ over $\Lambda$ is said to
  \emph{select} coordinate $(k,\ell) \in coords(T)$, denoted $T,
  (k,\ell) \models A$, if, there exists a run $\rho$ of $A$ on the
  event stream for $T$ starting from configuration $(1,1,1,q_0)$ with
  $q_0$ the initial state of $A$ such that there is a configuration
  $c$ in $\rho$ with $c = (i,k,\ell, \varsigma)$ for some $i$ and
  $\varsigma$ such that $\varsigma \in F$. (In particular, $\varsigma$
  is not suspended.)
\end{definition}
We write
$\sem{A}_T$ for the set of all coordinates selected by $A$ on $T$.

\begin{example}\upshape
  The following CA selects the same coordinates as region expression $\down^* \right^+ (\text{\texttt{ARUBA}} )$
\begin{center}
\begin{tikzpicture}[>=stealth,thick,baseline=0]
  \node (init)  at (0,0)   [state,initial] {\ };
  \node (a)     at (2,0)   [state] {\ };
  \node (final) at (4,0)   [state,accepting] {\ };

  \draw[->] (init) -- node[above] {$[$\textnormal{\texttt{ARUBA}}$]$} (a);
  \draw[->] (a)    -- node[above] {$\right$} (final);
  
  \path (init) edge[loop above] node {$\newrow$} (init);
  \path (init) edge[loop below] node {$[\true]$} (init);
  \path (a) edge[loop above] node {$\down$} (a);
  \path (final) edge[loop above] node {$\right$} (final);

                      \end{tikzpicture}
\end{center}
\hfill $\blacksquare$
\end{example}

Note that, in the definition above, a CA always selects coordinates
when started from the root coordinate (1,1) of the table (i.e., the
beginning of the event stream). In what follows, it will be convenient
for technical reasons to say that a $CA$ $A$ selects coordinate
$(k,\ell)$ in $T$ when started at coordinate $(k',\ell')$. This is formally
defined as follows. 

\begin{definition}
  Let $s = \sigma_1 \cdots \sigma_n$ be the event stream of $T$. Let
  $\sigma_i$ with $1 \leq i \leq n$ be the symbol in this stream that
  corresponds to the cell in $T$ with coordinate
  $(k',\ell')$.\footnote{Here, we tacitly make the convention that
    \newrow events do not corresponds to any cell, so $\sigma_i$ is an
    event of the form $\event{\Gamma}$.} Then
  $A$ selects $(k,\ell)$ when started from $(k',\ell')$ in $T$ if
  there exists a run $\rho$ of $A$ on $\sigma_1 \cdots \sigma_n$
  starting with configuration $(i,k',\ell', q_0)$ such that there is a
  configuration $c$ in $\rho$ with $c = (j,k,\ell,\varsigma)$ for some
  $j$ and some $\varsigma \in F$. (In particular, $\varsigma$ is not
  suspended.)
\end{definition}
We write $\sem{A}_{T,c}$ for the set of all coordinates selected by $A$ on
$T$ when started from coordinate $c$.

\begin{definition}
  Coordinate automaton $A$ \emph{expresses} coordinate expression
  $\varphi$ if $\sem{\varphi}_T = \sem{A}_T$, for every tokenized
  table $T$. We say that $\varphi$ is \emph{definable} by means of a
  CA if there exists a CA that expresses $\varphi$. Similarly, if
  $\alpha$ is a navigational expression, then $A$ 
  \emph{expresses } $\alpha$ if $\alpha(\{c\}) = \sem{A}_{T,c}$, for
  every tokenized table $T$ and every $c \in \coords(T)$.
\end{definition}

It should be noted that the fact that a coordinate automaton can use
coordinate expressions as oracles makes them very powerful. Indeed,
they are almost trivially able to express any coordinate or
navigational expression. In order to obtain
Theorem~\ref{theo:weak-stream}, however, we will compile coordinate
expressions into CA whose oracles are restricted. By restricted here
we mean that the oracles are of a different \emph{level} than
coordinate or navigational expressions that are being expressed. The
level intuitively corresponds to the maximum nesting level of
coordinate and navigational expressions. The formal definition is as
follows.
\begin{align*}
  \level(a) & = 0 \\
  \level(\rootcoord) & = 0 \\
  \level(\true) & = 0 \\
  \level(\varphi \wedge \psi) &=  \max(\level(\varphi),\level(\psi)) \\
  \level(\varphi \vee \psi) & = \max(\level(\varphi),\level(\psi))  \\
  \level(\neg \varphi) & = \level(\varphi) \\
  \level(\alpha(\varphi)) & = 1 + \max(\level(\alpha),
  \level(\varphi))\\
\end{align*}
\begin{align*}
\level(\varepsilon) & = 0 \\
  \level(\down) & = 0 \\  
  \level(\right) & = 0 \\  
  \level([\varphi]) & = 1 + \level(\varphi) \\  
  \level(\alpha \cdot \beta) & = \max(\level(\alpha), \level(\beta)) \\  
  \level(\alpha + \beta) & = \max(\level(\alpha), \level(\beta)) \\  
  \level(\alpha^*) & = \level(\alpha) \\  
\end{align*}

To illustrate, $\level(\down \cdot \right^*) = 0$, $\level([a] \cdot
\right) = 1$, and $\level(\right (a \wedge b) ) = 1$.

Define the \emph{level} of a CA to be the maximum level of any
coordinate expression occurring in it. We now establish a number of
technical lemmas that relate CA to coordinate and navigational expressions.

\begin{lemma}
  \label{lem:ca-vs-navexpr}
  Every forward navigational expression $\alpha$ is definable by means
  of a CA of level $\max(0,\level(\alpha)-1)$.
\end{lemma}
\begin{proof}
  We construct, by induction on $\alpha$, a CA $A$ that expresses
  $\alpha$ and is of level at most $\level(\alpha)-1$. The
  construction is essentially the same as the Thompson construction
  for transforming regular expressions into finite state automata.
  \begin{itemize}
  \item If $\alpha = \varepsilon$, then $A$ is the CA with one state,
    which is both initial and final, and which does not have any
    transitions. 
  \item If $\alpha = \down$, then $A$ is the CA with states $q$ and
    $q'$, where $q$ is initial and $q'$ is final, with the single
    transition $(q, \down, q')$.
  \item If $\alpha = \right$, then $A$ is the CA with states $q$ and
    $q'$, where $q$ is initial and $q'$ is final, with the single
    transition $(q, \right, q')$.
  \item If $\alpha = [\varphi]$, then $A$ is the CA with states $q$
    and $q'$ where $q$ is initial and $q'$ is final, with the single
    transition $(q, [\varphi], q')$.
  \item If $\alpha = \alpha_1 \cdot \alpha_2$, then let $A_1$ be the
    CA constructed by induction for $\alpha_1$, and $A_2$ the CA
    constructed by induction for $\alpha_2$. Take $A = A_1 \cdot A_2$,
    the CA obtained by the usual concatenation construction on $A_1$
    with $A_2$. (That is, we take the disjoint union of $A_1$ and $A_2$,
    renaming states where necessary, and link the final states of
    $A_1$ to the initial state of $A_2$ by means of an
    $\epsilon$-transition. The initial state is the initial state of
    $A_1$; the final states are the final states of $A_2$.)
  \item If $\alpha = \alpha_1 + \alpha_2$, then let $A_1$ be the
    automaton constructed by induction for $\alpha_1$ and $A_2$ the automaton
    constructed for $\alpha_2$. Then take $A$ to $A_1 + A_2$, the
    automaton obtained by performing the usual union construction on
    automata. (Take their disjoin union, renaming states where
    necessary, add a new initial state and add epsilon transition from
    this state to the initial states of $B_1$ and $B_2$,
    respectively. The final states consists of the final states of
    $B_1$ and $B_2$.)
  \item If $\alpha = \beta^*$, then let $B$ be the automaton created
    for $\beta$. Let $A$ be the automaton we obtain by adding an
    $\epsilon$-loop from the final states of $B$ to its initial
    state. 
  \end{itemize}
  In all cases, it is now routine to check that $A$ defines $\alpha$
  and that $\level(A)  \leq \max(0,\level(\alpha) - 1)$.
\end{proof}

\begin{lemma}
  \label{lem:ca-vs-cexpr-alpha-phi}
  Every forward coordinate expression  of the form $\alpha(\varphi)$ 
   is definable by means of a CA of level at most
  $\level(\alpha(\varphi)) - 1$.
\end{lemma}
\begin{proof}
  Observe that $\alpha(\varphi)$ is equivalent to $\beta(\rootcoord)$
  where $\beta$ is $(\right + \down)^* \cdot [\varphi] \cdot
\alpha$. By Lemma~\ref{lem:ca-vs-navexpr}, there exist a CA $A$
defining $\beta$ of level $\level(\beta) - 1 = \max(\level(\varphi) +
1, \level(\alpha)) -1 = \max(\level(\varphi), \level(\alpha) - 1) \leq
\level(\alpha(\varphi)) - 1$. Now observe that, since $\alpha(\varphi)
\equiv \beta(\rootcoord)$, it immediately follows that $A$ defines
$\alpha(\varphi)$.  (Recall that a CA expresses a CA if it gives the
same result when evaluation starts from root coordinate (1,1).)
\end{proof}

\begin{lemma}
  \label{lem:ca-vs-cexpr}
  Every forward coordinate expression is definable by a CA of the same
  level.
\end{lemma}
\begin{proof}
  Note that every forward coordinate expression $\varphi$ is
  equivalent to $\beta(\rootcoord)$ where $\beta = (\right + \down)^*
  \cdot [\varphi])$. By Lemma~\ref{lem:ca-vs-navexpr}, there exists a
  CA $A$ defining $\beta$ of level $\level(\beta) - 1 =
  \level(\varphi)$.  Now observe that, since $\varphi \equiv
  \beta(\rootcoord)$, it immediately follows that $A$ defines
  $\varphi$.  (Recall that a CA expresses a CA if it gives the same
  result when evaluation starts from root coordinate (1,1).)
\end{proof}

{\bf Convention.} In what follows, by Turing Machine we understand a
Turing Machine with a read-only input tape, a read-write work tape,
and a write-only output tape where the cursor on the input tape can
only advance to the right (never go left). We say that Turing machine
$M$ \emph{implements} CA $A$ if, for every tokenized table $T$ of dimension
$n \times m$ and event stream $s = \sigma_1 \cdots \sigma_k$ of
$T$ it is the case that $M$ outputs on its output tape the coordinates
of the cells selected by $A$ on $\sigma$.

\begin{proposition}
  \label{prop:ca-tm-logspace}
  For every CA $A$ there exists a Turing Machine $M$ that implements
  $A$ and that uses at most $O(m \log(m) + \log(n))$ space on its work
  tape, where $m$ is the number of columns and $n$ is the number of
  rows in the input table event stream.
\end{proposition}
\begin{proof}
  Let $A = (Q, q_0, F, \delta)$. Let $\Omega$ be the set of all
  oracles used in transitions of $A$, i.e., $\Omega = \{ \varphi \mid
  (q, [\varphi], q') \in \Delta, q,q'\in Q\}$. We construct $M$ by
  induction on the level $\lambda$ of $A$.

  \smallskip \noindent
  {\bf Base case. } If the level $\lambda$ of $A$ is $0$, then $M$
  operates as follows. It processes the event stream on its input tape
  from left to right, one event at a time, starting at the first
  event. During the processing it maintains on its work tape a tuple
  $(k,\ell,\act, \susp, \Phi)$, where $(k,\ell)$ is the coordinate
  corresponding to the current event being processed, $\act$ is the
  set of all active states that $A$ can be in any run of $A$ after
  having processed the events so far, $\susp$ is the set of all
  suspended states that $A$ can be in in any run after having
  processed the events so far, and $\Phi$ is the set of all oracles in
  $\Omega$ that select the current coordinate $(k,\ell)$. As such, $M$
  simulates all possible runs of $A$ on the input, similarly to how
  one normally simulates a non-deterministic finite state automaton by
  tracking all possible runs at once.

  In particular, $M$ starts with the tuple $(1,0, \{q_0\}, \emptyset,
  \emptyset)$ on its work tape before processing any event. When
  processing the next event, $M$ checks whether this is of the form
  $\event{\Gamma}$ or $\newrow$. If it is of the form $\event{\Gamma}$
  then it:
  \begin{compactenum}[1)]
  \item Increments $\ell$;
  \item Computes $\Phi$ for the new coordinate $(k,\ell)$. Note that this
    can be done using only the information in $\Gamma$ (the tokens of
    the current event) and the coordinate $(k,\ell)$. Indeed: since all
    oracles $\varphi \in \Omega$ are forward and level 0, they cannot
    contain subexpressions of the form $\langle \alpha \rangle$ or
    $\alpha(\psi)$. Therefore, each such $\varphi$ is a boolean
    combinations of literals, where each literal is either (1) a token
    $a$, (2) the constant $\true$, or (3) the root coordinate
    $\rootcoord$. Checking (2) is trivial, whereas checking (1)
    amounts to checking whether $a \in \Gamma$ and (3) amounts to
    comparing the coordinate $(k,\ell)$ with (1,1).
  \item It replaces $\act$ by $\{ q'\mid q \in \act, (q,\right,q') \in
  \delta\} \cup \{ q \mid q\colon \ell \in \susp\}$. It then adds to 
  $\act$ all states $p$ that can be reached from a state in this new $\act$ by
  traversing only $\epsilon$-transitions or $[\varphi]$-transitions,
  with $\varphi \in \Phi$.
\item Finally, it updates $\susp$ to \[ \{ q' \colon \ell \in Q \mid q \in
  \act, (q,\down,q') \in \Delta \}.\]
  \end{compactenum}
  If the current event is $\newrow$ then it :
  \begin{compactenum}[1)]
  \item increments $k$ and resets $\ell$ to $1$;
  \item moves to the next event on the input tape, which must be of
    the form $\event{\Gamma}$;\footnote{since we assume that all
      tables have at least one column.}
  \item Computes $\Phi$ for the new coordinate $(k,\ell)$.
  \item Replaces $\act$ by $\{q'\mid q \in \act, (q,\newrow, q') \in
    \delta\} \cup \{q' \mid q'\colon 1 \in \susp\}$. It then adds to
    $\act$ all states $p$ that can be reached from a state in this new
    $\act$ by traversing only $\epsilon$-transitions or
    $[\varphi]$-transitions, with $\varphi \in \Phi$.
\item Finally, it updates $\susp$ to \[ \{ q' \colon \ell \in Q \mid q \in
  \act, (q,\down,q') \in \Delta \}.\]
  \end{compactenum}
  
  After updating $(k,\ell,\act,susp,\Phi)$, $M$ checks if $\act \cap F
  \not = \emptyset$. If so, it outputs $(k,\ell)$.
  
  Observe that the space required by $M$ is:
  \begin{itemize}
  \item $\log(n) + \log(m)$ bits to store the coordinate $(k,\ell)$;
  \item at most $|Q|$ bits for storing $\act$;
\item at most $|Q| m \log(m)$ bits for storing $\susp$;
\item at most $|\Omega|$ bits for storing $\Phi$
  \end{itemize}
  Hence, since $|Q|$ and $|\Omega|$ are constant, $M$ runs in space
  $O(m \log(m) + \log(n))$, as desired.

  \smallskip \noindent
  {\bf Induction step.} If the level $\lambda$ of $A$ is $> 0$, then
  let $N$ be the set of all coordinate expressions of the form
  $\alpha(\varphi)$ that occur as a subexpression of some oracle in
  $\Omega$. By definition, each of these $\alpha(\varphi)$ is of level at
  most $\lambda$. By Lemma~\ref{lem:ca-vs-cexpr-alpha-phi}, there hence exists
  for each such $\alpha(\varphi)$ a CA $A_{\alpha(\varphi)}$ of level
  at most $\lambda - 1$ that expresses it. By induction hypothesis,
  there hence exists, for each $\alpha(\varphi) \in N$, a Turing
  Machine $M_{\alpha(\varphi)}$ that implements $\alpha(\varphi)$ in
  space $O(m \log(m) + log(n))$. 

  We then construct the Turing Machine $M$ for $A$ as follows. For
  ease of exposition, $M$ will have multiple (but a fixed number of)
  work tapes. Since each of these will use only $O(m \log(m) + \log(n))$ cells,
  it is standard to transform $M$ in a single-tape Turing Machine that
  runs in space $O(m \log(m) + \log(n))$.

  In particular, $M$ has a $|N| + 1$ work tapes: a principal work tape
  and an auxiliary work tape for each $\alpha(\varphi) \in N$. During
  processing, $M$ simulates $A$ on its principal work tape, and, in
  parallel, the Turing Machine for $M_{\alpha(\varphi)}$ on the
  auxiliary tape for $\alpha(\varphi)$. During the simulation of
  $M_{\alpha(\varphi)}$ we take care to never construct any output,
  but merely checks whether the current cell on the input tape should
  be output according to $M_{\alpha(\varphi)}$.

  The simulation of $A$ on its principal work tape happens in the
  exact same way as for the case where $\lambda = 0$. That is, we
  maintains a tuple $(k,\ell,\act,\susp, \Phi)$ for $A$ that simulates
  all possible runs of $A$. The only difference is that when we update
  this tuple in response to reading a new event from the input, we
  first update all the auxiliary work tapes, and then compute the set
  $\Phi \subseteq \Omega$ of all of $A$'s oracles that select the
  current cell $(k,\ell)$ as follows.

  Since each $\psi \in \Omega$ is forward, each $\psi$ is a boolean
  combination of \emph{literals}, where each literal is either (1) $a$
  with $a \in \tokens$; (2) $\true$; (3) $\rootcoord$; or (4)
  $\alpha(\varphi)$. Cases (1)--(3) can be checked as before. We can
  check whether current coordinate $(k,\ell)$ is selected by
  $\alpha(\varphi)$ simply by looking at the simulation of
  $M_{\alpha(\varphi)}$ on the work tape for $\alpha(\varphi)$ and
  verify whether $M_{\alpha(\varphi)}$ would output the current
  coordinate. As such, we can easily compute at any given instant
  whether the current coordinate is selected by $\psi$.

  As before, after the update of $(k,\ell,\act,\susp,\phi)$, $M$ checks
  whether $\act \cap F \neq \emptyset$ and, if so, writes $(k,\ell)$ on
  its output tape.

  Now note that, as before, $M$ uses $O(m \log(m) + \log(n))$ space on
  its principal work tape, and (by induction hypothesis) $O(m \log (m)
  + log(n))$ on each of its auxiliary work tapes. It hence uses $O(m
  \log(m) + log(n))$ space in total, as desired.
\end{proof}

From this, we derive Theorem~\ref{theo:weak-stream} as follows.

\begin{proof}[of Theorem~\ref{theo:weak-stream}]
  Let $R$ be a tabular schema, i.e. a set of rules of the form
  $\varphi \to e$ with $\varphi$ a coordinate expression and $e$ a
  content expression.

  By Lemma~\ref{lem:ca-vs-cexpr} every forward coordinate expression
  can be expressed by means of a CA which, by
  Proposition~\ref{prop:ca-tm-logspace}, can be evaluated by a Turing
  Machine in space $O(m \log(m) + \log(n))$, where $m$ is the number
  of columns in the input, and $n$ the number of rows.

  From this, we construct a Turing Machine $M$ that validates its
  input event stream w.r.t.\ $R$ as follows. $M$ has a fixed number of
  work tapes. In particular, for each coordinate expression $\varphi$
  serving as the left-hand side of a rule in $R$, $M$ has one tape on
  which it simulates the Turing Machine that evaluates
  $\varphi$. Here, $M$ prevents any output that may be generated by
  the Turing Machine for $\varphi$; but records when this machine
  would be doing so.

  All of these left-hand-sides are simulated in parallel upon reading
  the event stream. In addition, for each right-hand side $c$, $M$ has
  an auxiliary tape on which it simulates a finite state automaton for
  $c$. For each rule $\varphi \to c$, and in each position in the
  event stream of the form $\event{\Gamma}$, whenever it finds that
  $\varphi$ would select the current cell, $M$ simulates reading
  $\Gamma$ in the NFA for $c$: it proceeds from the current state set
  for $c$ according to all $\tau \in \Gamma$. 

  Under the row-based semantics, whenever we encounter a $\newrow$,
  each $c$ must be in a finite state, otherwise the input is invalid
  w.r.t.\ $R$. Whenever we see $\newrow$, we move each of the NFAs for
  $c$ back to their initial state.

  Under the region-base semantics, each $c$ has to be in a final state
  at the end of the input (and we never need to move $c$ back to its
  initial state upon $\newrow$).
\end{proof}

\makeatletter{}\section{Proof of Theorem~\ref{thm:forward-guarded-strong-streaming}}

\textsc{Theorem~\ref{thm:forward-guarded-strong-streaming}}. {\it 
  Guarded forward core-\chisel schemas are strongly streamable.}

\begin{proof}[sketch]
  We can use the algorithm of Theorem~\ref{theo:weak-stream} for weak
  streamability with the additional observation that, for guarded
  schemas, the columns that need to be remembered when going from one
  row to the next are independent of the width of the table $T$. More
  precisely, when going from one row to the next, we can store for
  each subexpression $\varphi$ of a region selection expression a set
  of pairs $P_\varphi$ of the form $(c,e)$ where $c$ is a column
  number and $e \in \{=,\geq\}$. The semantics is that, on the next
  row $r$, we need to continue the evaluation of this subexpression in
  the cells $\{(r,c) \mid {(c,=)} \in P_\varphi\} \cup \{(r,i) \mid
  {(c,\geq)} \in P_\varphi$ and $i \geq c\}$. The size of each such set
  $P_\varphi$ can be bounded by $2^{O(|\varphi|)}$, which is
  formalized in Lemma~\ref{lem:strong-streaming-columns}.
\end{proof}

For a coordinate expression $\varphi$, table $T$, and a row number $r$ we denote
by $\cells{\varphi}{T,r}$ the cells of $\sem{\varphi}_T$ in row
$r$. That is, $\cells{\varphi}{T,r} = \sem{\varphi}_T \cap \{(r,k)
\mid k \in \nat\}$.

In the next lemma, a \emph{right-open interval} (on a row $r$) is a set
of cells $S$ for which there exists a $k \in \nat$ such that $S =
\{(r,i) \mid i \geq k\}$.

\begin{lemma}\label{lem:strong-streaming-columns}
  \begin{enumerate}[(a)]
  \item If a coordinate expression $\varphi$ is row-guarded then, for
    each table $T$ and row coordinate $r$, the set $\cells{\varphi}{T,r}$
    consists of $2^{O(|\varphi|)}$ cells plus, optionally, a
    right-open interval on $r$.
  \item If a coordinate expression $\varphi$ is guarded then, for each
    table $T$ and row coordinate $r$, the set
    $\cells{\down^*(\varphi)}{T,r}$ consists of $2^{O(|\varphi|)}$
    cells plus, optionally, a right-open interval on $r$.
 \end{enumerate}
\end{lemma}
\begin{proof}
  The lemma is proved by a straightforward induction on forward
  coordinate and navigational expressions. We provide it in full
  detail for the sake of completeness.

  For a coordinate expression $\varphi$ and a row $r$, we denote by
  $\numbercells{\varphi}{r}$ the number of coordinates of
  $\sem{\varphi}$ in row $r$, that is, the number of elements in
  $\{(r,k) \in \sem{\varphi} \mid k \in \nat\}$ if it is finite, and
  $\infty$ otherwise.

  \noindent (a) We prove this case by induction on the definition of
  guardedness in forward coordinate expressions.  The induction base cases for
  coordinate expressions $\varphi$ are $\varphi = a$, $\varphi =
  \textsf{root}$, $\varphi = \true$, and  $\varphi =
  \right^*(\psi)$.

  In the first case, if $\varphi = a$ is row-guarded then it only
  appears at most once per row by definition of the predicate {\tt
    unique-per-row}. Therefore, for each row $r$,
  $\numbercells{\varphi}{r} \leq 1$ and (a) is fulfilled.

  In the second case, if $\varphi = \textsf{root}$, then
  $\sem{\varphi}$ only contains a single cell. Again, for each row $r$,
  $\numbercells{\varphi}{r} \leq 1$ and (a) is fulfilled.

  In the third case, if $\varphi = \true$, then, for each row $r$,
  $\numbercells{\varphi}{r} = \infty$. Moreover,
  $\cells{\varphi}{r} = \{(r,k) \mid k \geq 1\}$, which is a
  right-open interval. Again, (a) is fulfilled.

  Fourth, when $\varphi = \right^*(\psi)$, we have that $\psi$ is a
  boolean combination of $\textsf{root}$, $\true$, and
  tokens. Therefore, it can be decided whether a cell $c$ is in
  $\sem{\psi}$ by looking at the predicates of $c$ and by testing
  whether it is cell $(1,1)$ or not. Such tests can be made by
  inspecting $c$ alone. For each row $r$, we either have that
  $\cells{\psi}{r}$ is empty or not.  If it is empty, then
  $\cells{\varphi}{r}$ is also empty and (a) follows. If it is
  non-empty, then $\cells{\varphi}{r} = \{(r,j) \mid j \geq i\}$, where
  $i$ is minimal such that $(r,i) \in \cells{\psi}{r}$. Since this is a
  right-open interval, (a) follows.

  For the inductive step, we consider the cases $\varphi = \psi_1 \lor
  \psi_2$, $\varphi = \psi_1 \land \psi_2$, $\varphi = \varepsilon(\psi)$, $\varphi = \down(\psi)$,
$\varphi = \right(\psi)$, $\varphi = \down(\alpha(\psi))$,
$\varphi = \right(\alpha(\psi))$, and $\varphi =
(\alpha+\beta)(\psi)$.

  In the first case, for each row $r$, we have
  that $\cells{\varphi}{r} = \cells{\psi_1}{r} \cup
  \cells{\psi_2}{r}$. Here (a) follows immediately from the inductive
  hypothesis and the observation that 
  the union of two right-open intervals is again a right-open interval.

  The second case is analogous to the first, but we observe that also
  the intersection of two right-open intervals is again a right-open
  interval.

  Third, when $\varphi = \varepsilon(\psi)$ case (a) follows
  immediately by induction and the observation that
  $\cells{\varepsilon(\psi)}{r} = \cells{\psi}{r}$ for each $r$.

  Cases four to seven are similar to case three. For example, when
  $\varphi = \down(\psi)$, then, for $r = 1$ we have
  $\cells{\varphi}{r} = \emptyset$ and for $r > 1$ we have
  $\cells{\varphi}{r} = \{(i,r) \mid (i,r-1) \in
  \cells{\psi}{r-1}\}$. Again (a) follows by induction.

  In the last case we have $\varphi = (\alpha+\beta)(\psi)$ and we
  already know that $\alpha(\psi)$ and $\beta(\psi)$ are
  row-guarded. Here, we again have that $\cells{\varphi}{r} =
  \cells{\alpha(\psi)}{r} \cup \cells{\beta(\psi)}{r}$. Here, (a)
  follows by induction in a similar way as the first inductive
  case. However, we need to take a little bit more care about the
  finite part of $\cells{\varphi}{r}$. Assume \mbox{w.l.o.g.} that
  $|\alpha| \geq |\beta|$ (if not, the roles of $\alpha$ and $\beta$
  can be interchanged in the following). Then, the finite part of
  $\cells{\varphi}{r}$ consists of $2^{O(|\alpha|+|\psi|)} +
  2^{O(|\beta|+|\psi|)}$ cells, which is at most $2 \cdot
  2^{O(|\alpha|+|\psi|)}$ cells and which, in turn, is bounded from
  above by $2^{O(|\varphi|})$ cells, which also proves (a) in this case.
  
  This concludes the proof of Lemma~\ref{lem:strong-streaming-columns}(a).
  
  \noindent (b) We again proceed by induction on the definition of
  guardedness in forward coordinate expressions. The induction base cases for
  coordinate expressions $\varphi$ are $\varphi = a$, $\varphi =
  \textsf{root}$, $\varphi = \true$, and $\varphi =
  \right^*(\psi)$.

 In the first case, if $\varphi = a$ is guarded then it only
  appears at most once in the table by definition of the predicate {\tt
    unique}. If $\sem{\varphi}$ is empty, then $\down^*(\varphi)$ is
  empty in which case (b) is fulfilled. If not, we have that
  $\sem{\varphi} = \{(i,j)\}$ for some $i$ and $j$. Here, we have that
  $\sem{\down^*(\varphi)} = \{(k,j) \mid k \geq i\}$, which fulfils
  the conditions of case (b).

  In the second case, if $\varphi = \textsf{root}$, then
  $\sem{\varphi}$ only contains a single cell. Here, we have that
  $\sem{\down^*(\varphi)} = \{(k,1) \mid k \geq 1\}$, which fulfils (b).

  In the third case, if $\varphi = \true$, then,
  $\sem{\down^*(\varphi)}$ contains every cell in the table. Clearly,
  this is a right-open interval for each row and (b) is fulfilled.

  In the fourth case, if $\varphi =  \right^*(\psi)$ we can decide for
each individual cell $c$ whether $c \in \sem{\varphi}$ by only
inspecting $c$, analogously as in (a). Also analogously, For each row
$r$, we either have that $\cells{\varphi}{r}$ is empty or a right-open
interval of the form $\cells{\varphi}{r} = \{(r,j) \mid j \geq i\}$, where
  $i$ is minimal such that $(r,i) \in \cells{\psi}{r}$. Therefore,
  $\cells{\down^*(\varphi)}{r}$ is also empty or a right-open interval
  on each row $r$.

  For the inductive step, we consider the cases $\varphi = \psi_1 \lor
  \psi_2$, $\varphi = \psi_1 \land \psi_2$, $\varphi = \varepsilon(\psi)$, $\varphi = \down(\psi)$,
  $\varphi = \right(\psi)$, $\varphi = \down(\alpha(\psi))$,
  $\varphi = \right(\alpha(\psi))$, $\varphi =
  (\alpha+\beta)(\psi)$, $\varphi = \down^* \cdot \alpha(\psi)$, and
  $\varphi = \right^*\cdot\alpha(\psi)$.
  
  In the first case, we have that $\down^*(\psi_1)$ and
  $\down^*(\psi_2)$ are guarded by induction, so they fulfil condition
  (b). Furthermore, we have that $\sem{\down^*(\psi_1 \lor \psi_2)} =
  \sem{\down^*(\psi_1)} \cup \sem{\down^*(\psi_2)}$. This means that,
  for each row $r$, we also have that $\cells{\varphi}{r} =
  \cells{\psi_1}{r} \cup \cells{\psi_2}{r}$. Here (b) follows
  immediately from the inductive hypothesis and the observation that the
  union of two right-open intervals is again a right-open interval.

  The second case is analogous to the first, but takes intersections
  of right-open intervals instead of unions.

  Case three, where $\varphi = \varepsilon(\psi)$, (b) follows
  immediately from the guardedness of $\psi$ and the observation that
  $\cells{\varepsilon(\psi)}{r} = \cells{\psi}{r}$ for each $r$.

  Cases four to eight are similar to case three. In case eight, where
  $\varphi = (\alpha+\beta)(\psi)$, we can bound the number of cells
  per row of $\sem{\down^*(\alpha+\beta)(\psi)}$ that are not in the
  right-open interval in exactly the same way as for the analogous
  case in (a).

  In case nine, we have that $\varphi = \down^* \cdot \alpha(\psi)$
  and $\alpha(\psi)$ is guarded. This case immediately follows from
  the induction hypothesis since $\sem{\down^* \varphi} = \sem{\down^*
    \down^* \alpha(\psi)} = \sem{\down^* \alpha(\psi)}$.

  Finally, in the last case, we have $\varphi =
\right^*\cdot\alpha(\psi)$, where $\alpha(\psi)$ is row-guarded. Due
to the row-guardedness of $\alpha(\psi)$, we know by (a) that
$\sem{\alpha(\psi)}$  for each row $r$ consists of $2^{O(|\varphi|)}$ cells in $r$ plus,
    optionally, a right-open interval on $r$. Therefore,
    $\right^*\cdot\alpha(\psi)$ on a row $r$ is either empty, or a
  right-open interval. From this, we immediately have that also
  $\down^*\right^*\alpha(\psi)$, on each row $r$, is either empty or a
right-open interval. This concludes the proof of Lemma~\ref{lem:strong-streaming-columns} (b).
\end{proof}

\end{document}